\documentclass[journal]{IEEEtran}
\usepackage[T1]{fontenc}
\usepackage{cite}
\usepackage{amsmath,amssymb,amsfonts}
\usepackage{amsthm}
\usepackage{algorithm}
\usepackage{algpseudocode}
\usepackage{graphicx}
\usepackage{textcomp}
\usepackage{xcolor}
\usepackage{colortbl}
\usepackage{booktabs}
\usepackage{multirow}
\usepackage{array}
\usepackage{stfloats}
\usepackage{url}
\usepackage[hidelinks]{hyperref}

\newcommand{\method}{\textsc{SCC}}
\newcommand{\yhat}{\hat{\mathbf{y}}}
\newcommand{\yb}{\mathbf{y}}
\newcommand{\xb}{\mathbf{x}}
\newcommand{\svec}{\mathbf{s}}
\newcommand{\mvec}{\mathbf{m}}

\newtheorem{proposition}{Proposition}

\definecolor{cblue}{RGB}{31,119,180}
\definecolor{cgreen}{RGB}{44,160,44}
\definecolor{corange}{RGB}{255,127,14}
\definecolor{cred}{RGB}{214,39,40}
\definecolor{cgray}{RGB}{110,110,110}

\begin{document}

\title{Severity-Channel Conditioning: An Injection-Level Study
for Consumer-Device Deblurring}

\author{Abdullah~Al~Shafi,
        Kazi~Saeed~Alam,
        Sumaiya~Rahim~Suma,
        and~Sk.~Imran~Hossain,~\IEEEmembership{Member,~IEEE}%
\thanks{This work has been submitted to the IEEE for possible publication. Copyright may be transferred without notice, after which this version may no longer be accessible. \textit{(Corresponding author: Sk. Imran Hossain.)}}%
\thanks{The authors are with the Department of Computer Science and
Engineering, Khulna University of Engineering \& Technology, Khulna~9203,
Bangladesh (e-mail: abdullah.shafi99@gmail.com; saeed.alam@cse.kuet.ac.bd;
sumaiya.rahim234@gmail.com; imran@cse.kuet.ac.bd).}}

\markboth{Preprint}%
{Shafi \MakeLowercase{\textit{et al.}}: Severity-Channel Conditioning: An Injection-Level Study}

\maketitle

\begin{abstract}
Motion-deblurring models lose substantial accuracy when transferred from professional
cameras to smartphone-derived imagery. This paper studies whether a severity
descriptor, partly derived from the paired sharp target, helps adaptation on the
iPhoneBlur consumer-device benchmark. Under a shared update budget, input concatenation of a severity map is compared with global feature modulation by capture metadata, spatial feature modulation by the severity map, and a tier-contrastive output objective, together with an unconditioned spectral-loss control, all fine-tuned from the same domain-adapted latent-diffusion baseline. The
primary experiment is an oracle diagnostic: target-derived information is supplied
at inference. For the evaluated model pair, input concatenation improves overall
peak signal-to-noise ratio from $31.23$ to $31.30$~dB, whereas feature modulation, the contrastive objective, and the spectral-loss control do not exceed the baseline. The gain is positive on
$87.2\%$ of test images and on $14$ of $15$ held-out source videos; this video-level
analysis describes the evaluated checkpoint pair and does not establish robustness to
retraining, and three training seeds are reported separately. The input-level gain
recurs on a transformer backbone. In one exploratory run, zeroing the channel at
inference retains most of the gain in peak signal-to-noise ratio but lowers
structural similarity below the baseline. A normalized-step diagnostic does not
support the expected larger-gradient explanation. The results characterize the
evaluated configurations and do not establish injection location as the sole cause,
a spatial-locality mechanism, or a reliable deployment advantage.
\end{abstract}


\section{Introduction}
\IEEEPARstart{M}{otion} deblurring is developed and benchmarked almost entirely on
professional capture hardware, action cameras and DSLR rigs, yet it is deployed
overwhelmingly on smartphones. This mismatch carries a concrete cost. On the
iPhoneBlur benchmark~\cite{iphoneblur}, a stratified consumer-device deblurring
dataset, restorers trained on the standard GoPro
corpus~\cite{nah2017gopro} and applied zero-shot to real iPhone captures lose
six to seven decibels of PSNR, several times larger than the two-to-three
decibel drops typical of cross-evaluation among professional
datasets~\cite{rim2020realblur,shen2019hide}. The gap reflects a domain shift: phone-derived captures differ from professional ones in image-signal-processing (ISP) and in handheld motion statistics. Fine-tuning on the target
domain recovers most of this gap; in the experiments reported here a latent-diffusion
backbone~\cite{diffir} rises from $27.22$~dB overall (GoPro zero-shot) to
$31.23$~dB, matching the strongest published methods~\cite{nafnet,fftformer}, but
one limitation remains. The fine-tuned model receives no explicit severity descriptor, although a mildly and a severely blurred frame differ by more than seven decibels in achievable quality~\cite{iphoneblur}; it must infer severity from the image alone. We refer to this absence of an explicit severity signal as
\emph{descriptor-blindness}; supplying that signal is the subject of this paper.

The iPhoneBlur benchmark is distinctive in supplying the signal that a
descriptor-blind model lacks. Every sample carries measured
capture metadata: a paired high-frequency-energy difference, a Farneb\"ack optical-flow~\cite{farneback} magnitude, a temporal blur window, a blur duration, and a sensor-noise estimate. Together these constitute a compact \emph{degradation descriptor} that
summarizes degradation that the fine-tuned model is not given explicitly, and our hypothesis is that supplying the descriptor to the restorer allows it
to condition its behaviour on each input and recover part of the residual that
domain fine-tuning leaves. Because the descriptor is computed from the captured video rather than from a gyroscope~\cite{gyro} or an event camera~\cite{eventdeblur}, it requires no additional capture hardware, although part of it depends on paired sharp frames. Its
deployment availability is nonetheless not automatic, and is examined explicitly in
Section~\ref{sec:prelim}. What the hypothesis leaves open is an architectural
question. Degradation-conditioned restoration is well
established~\cite{dasr,airnet,promptir,msdinet,ufpnet}, but existing methods each
commit to a single injection level chosen by design and do not compare levels under a
matched budget. A conditioning signal can enter at the input by concatenation, at the
feature level by modulating activations~\cite{film,sft}, at the output level by
shaping the prediction's representation~\cite{simclr}, or at the objective level
through the loss~\cite{focalfreq}. These levels perturb a strongly pretrained backbone in different ways.

The question is addressed empirically for selected configurations: input concatenation of a severity map, global and spatial feature modulation, a tier-contrastive output objective, and an unconditioned spectral-loss control, all compared against a
\emph{schedule-matched} baseline that receives the same additional fine-tuning
budget without the descriptor. With the descriptor supplied at inference, only
input-level conditioning, a severity channel appended to the image and fine-tuned
end-to-end, improves over the schedule-matched fine-tune, consistently across the
three difficulty tiers and on a clear majority of individual images; feature modulation, the contrastive objective, and the spectral-loss control do not. This oracle improvement is modest
($+0.07$~dB overall). Its uniform sign across tiers, an $87\%$ per-image win rate
that survives aggregation to the video level for the evaluated model pair
($p=0.0002$), and its recurrence on a transformer backbone show that it is not
confined to a few images or videos; with three training seeds, robustness to
retraining is not established. Relative to the direct blur-prior deblurring literature (Section~\ref{sec:related}), the contribution is not the first use of a degradation descriptor in deblurring, which is established, but a controlled comparison of the level at which such a descriptor enters. The result is accompanied by a function-preserving-initialization property of
the added channel, a statement of when a descriptor can lower attainable risk, and
a discussion of possible reasons for the ordering, which the experiments do not
isolate.

This paper makes three contributions. First, it formulates the injection-level
question for a measured degradation descriptor and compares selected adaptation
configurations on the same backbone under a matched update budget, including an
unconditioned loss control. Second, it reports that, with the descriptor supplied
at inference, only input-level severity-channel conditioning exceeds the
schedule-matched baseline, consistently across difficulty tiers and on a large
majority of individual images, while the evaluated feature- and output-level variants and the spectral-loss control do not. Single-run descriptor and capacity controls indicate
that the gain tracks the descriptor's content rather than the widened convolution. Third, it states the limits of
the result: the descriptor's high-frequency term is derived from paired data and is
not available in deployment, zeroing the channel at test time in a single run costs $0.02$~dB of
PSNR and lowers SSIM slightly below the baseline, and the per-image statistics are
computed over frames nested within $15$ held-out videos rather than independent captures. The study is therefore diagnostic: headline results use the target-derived descriptor at inference and characterize the evaluated configurations, not a deployed model.
The
remainder of the paper reviews related work by method family
(Section~\ref{sec:related}), formalizes the problem and the domain gap
(Section~\ref{sec:prelim}), defines the evaluated configurations
(Section~\ref{sec:diag}), develops the proposed method and its theory
(Section~\ref{sec:method}), presents the evaluation (Section~\ref{sec:exp}), and closes with discussion and conclusions
(Sections~\ref{sec:disc}--\ref{sec:concl}).

\section{Related Work}
\label{sec:related}
Learning-based motion deblurring replaced hand-designed priors with end-to-end
networks: multi-scale~\cite{nah2017gopro}, scale-recurrent~\cite{srn},
multi-stage~\cite{mprnet}, and multi-input multi-output~\cite{mimounet} designs;
activation-free~\cite{nafnet} and half-instance-normalized~\cite{hinet} backbones;
attention-based restorers such as Restormer~\cite{restormer},
Uformer~\cite{uformer}, and Stripformer~\cite{stripformer}; and frequency-domain
methods including FFTformer~\cite{fftformer}, residual Fourier
transformation~\cite{deeprft}, and focal frequency loss~\cite{focalfreq}. All of
them degrade sharply on consumer captures~\cite{iphoneblur}, and the present study
is orthogonal to the choice of architecture. Diffusion models reverse a gradual
noising process~\cite{ho2020ddpm}, have been applied to image-to-image
restoration~\cite{saharia2022palette}, and gain efficiency in a compressed
latent~\cite{rombach2022ldm}; DiffIR diffuses a compact degradation prior that
conditions a transformer generator~\cite{diffir}, and guidance-calibrated variants
have been compressed by distillation~\cite{dash}. How auxiliary information enters a
network is itself a design axis: FiLM predicts per-channel affine
parameters~\cite{film}, spatial feature transform makes them spatially
varying~\cite{sft}, contrastive objectives group like samples~\cite{simclr}, and
input concatenation leaves the architecture intact.

Degradation-conditioned restoration is established. Blind super-resolution injects
a contrastively learned degradation representation through degradation-aware
convolutions~\cite{dasr}; all-in-one restoration encodes the corruption
contrastively~\cite{airnet}, as learned prompts~\cite{promptir}, as a human
instruction~\cite{instructir}, as an input-level visual prompt~\cite{prores}, or as
a vision-language embedding used through cross-attention~\cite{daclip}. These
descriptors are learned from the degraded image to discriminate corruption type,
whereas the descriptor here is measured and grades the severity of one degradation.
For deblurring, spatial blur estimates enter at the feature level in
MSDI-Net~\cite{msdinet} and through dark- and bright-channel priors~\cite{dbcpenet},
at every encoder stage in UFPNet~\cite{ufpnet}, as a gating map in blur-pixel
discretization~\cite{blurpixel}, and in image space through motion-adaptive
filters~\cite{miscfilter} and kernel priors~\cite{paramblur}; ParamISP conditions
the imaging pipeline itself on EXIF metadata~\cite{paramisp}. Each method commits to
one injection level by design. Super-resolution work has compared input
concatenation with feature modulation for a fixed code~\cite{sft,srmd,ikc}, but none compares input-, feature-, and output-level configurations for motion deblurring under a matched budget with a privileged descriptor, alongside an unconditioned loss control. That
comparison is the contribution here; because the configurations do not share an
identical signal, it characterizes them rather than isolating level as the sole
factor. Gyroscope~\cite{gyro} and event~\cite{eventdeblur} sensors supply motion
priors but need extra hardware; camera-formation modelling requires linear-domain
operations~\cite{brooks2019unprocessing}, unpaired translation has bridged
synthetic and real data~\cite{cyclegan}, and realistic blur synthesis adapts the
training data for a single camera~\cite{rim2022rsblur}. The present study instead
conditions adaptation on a per-sample descriptor and requires no resynthesis.

\section{Problem Setting}
\label{sec:prelim}
Let $\xb\in\mathbb{R}^{3\times H\times W}$ be a blurred observation and $\yb$ the
sharp image. Motion blur integrates a moving scene over the exposure interval, and
on a physical device the result is further transformed by the imaging pipeline,
\begin{equation}
\xb = \Gamma\!\Big(\tfrac{1}{T}\!\int_{0}^{T}\!\Gamma^{-1}\!\big(\Phi_{\text{ISP}}(\mathcal{S}(t))\big)\,dt\Big) + \boldsymbol{\eta},
\label{eq:blurmodel}
\end{equation}
where $\mathcal{S}(t)$ is the instantaneous scene, $T$ the exposure,
$\boldsymbol{\eta}$ sensor noise, $\Phi_{\text{ISP}}$ the ISP that performs tone
mapping, HDR fusion, denoising, and sharpening, and $\Gamma$ the gamma encoding with
inverse $\Gamma^{-1}$. The operator order matters: in the
benchmark used here, each source frame is captured through the full smartphone video
pipeline \emph{first}, and the blurred observation is formed by averaging those
already-processed frames after gamma linearization~\cite{iphoneblur}. The blur is
therefore synthesized from high-frame-rate smartphone video rather than acquired as a
long-exposure photograph, and $\Phi_{\text{ISP}}$ acts per frame inside the integral,
not once on a clean temporal average. This composition is one plausible source of the consumer domain gap, although the
experiments here do not separate it from motion and scene statistics. The benchmark
reports a mean high-frequency difference of $23\%$ between paired sharp and
blurred frames~\cite{iphoneblur}; because it varies from capture to capture, two
frames with similar motion may call for different restoration.

Fine-tuning the backbone on the target domain, denoted the Domain-Adapted
Baseline (DAB), adapts to the target domain on average. The DAB has fixed learned parameters, $f_\theta:\xb\mapsto\yhat$, and receives no explicit external severity descriptor. Its response remains input-dependent, and it may infer degradation from the blurred pixels. The oracle descriptor can supply additional target-derived information; the experiments assess its value for the fitted models rather than establish a limitation of all image-only predictors. Proposition~\ref{prop:blind} states when this matters: a restorer observing only
$\xb$ can attain at best the risk $R_{\xb}$, while a restorer additionally observing
a descriptor $\mvec$ can attain $R_{\xb,\mvec}\le R_{\xb}$, with strict improvement
only where $\mvec$ changes the Bayes-optimal response. A descriptor that distinguishes such inputs can carry information unavailable to the DAB. The benchmark supplies such a descriptor for every sample: five measured metadata
fields, a paired high-frequency-energy difference $m_{\text{isp}}$
(blurred-minus-sharp high-frequency energy, negative when blurring removes high-frequency energy), an optical-flow motion magnitude, a
temporal blur window, a blur duration, and a sensor-noise estimate, collected into
a vector $\mvec$. \method{} uses a two-term severity map derived from $m_{\text{isp}}$ and the
image gradient alone (Section~\ref{sec:method}); the FiLM baseline receives all
five fields, giving the feature-level mechanism every advantage
(Section~\ref{sec:diag}). One implication matters at inference: $m_{\text{isp}}$ is computed from the sharp image and
is therefore ground-truth-derived at the benchmark level. At deployment, a proxy would be needed instead, either from the ISP where it exposes the statistic or
estimated from the blurred input alone. Section~\ref{sec:disc} reports a single-run diagnostic of this setting: zeroing the
severity map, which removes the descriptor rather than substituting a constant for
$m_{\text{isp}}$, costs $0.02$~dB ($31.30\!\rightarrow\!31.28$~dB, still above the
baseline's $31.23$~dB) and lowers SSIM slightly below the baseline. The question studied here is where $\mvec$ should enter a conditioned restorer $f_\theta(\xb,\mvec)$.
Success is measured in PSNR, SSIM~\cite{ssim}, and LPIPS~\cite{lpips}, stratified into
Easy, Medium, and Hard tiers, a stratification defined by the benchmark itself and
independent of any mechanism studied here, and per-model behaviour is summarized
by the difficulty gap
$\Delta_{\text{gap}}=\text{PSNR}_{\text{Easy}}-\text{PSNR}_{\text{Hard}}$.

\section{Conditioning Mechanisms}
\label{sec:diag}
The configurations share the backbone, dataset, optimizer, and update budget, and differ in where conditioning enters and what signal it receives. Figure~\ref{fig:4level} illustrates the levels considered.

\begin{figure*}[t]
\centering
\includegraphics[width=\textwidth]{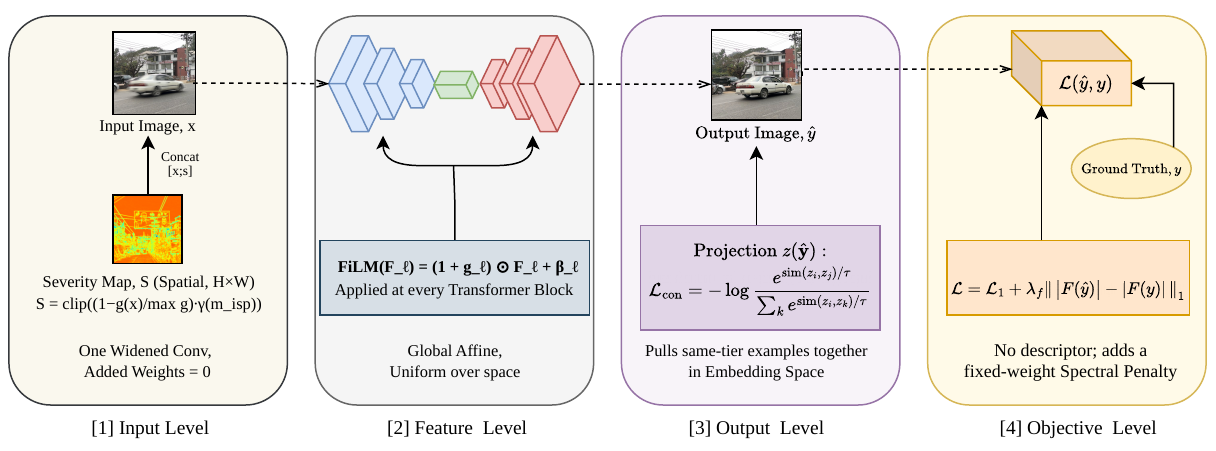}
\caption{Injection levels considered. Dashed arrows show the restorer's data path, from input image through features and prediction to the loss. Each configuration modifies only its own level, and no experiment combines mechanisms. Input: severity map concatenated to the image. Feature: affine modulation by the metadata vector (global) or the severity map (spatial). Output: a tier-contrastive term on a projection of the prediction. Objective: a fixed-weight spectral loss that receives no descriptor, reported as an unconditioned control.}
\label{fig:4level}
\end{figure*}

At the input level, Severity-Channel Conditioning (\method) converts $\mvec$ into a
spatial severity map $\svec\in[0,1]^{H\times W}$ and concatenates it to the RGB
image, so that the conditioned model computes $f_\theta([\xb;\svec])$. The map is
derived from the blurred image and the paired high-frequency statistic,
\begin{equation}
\svec = \mathrm{clip}_{[0,1]}\!\Big(\big(1-\tfrac{g(\xb)}{\max g(\xb)}\big)\cdot
\gamma(m_{\text{isp}})\Big),
\label{eq:sev}
\end{equation}
where $g(\xb)$ is a Gaussian-smoothed Sobel gradient-magnitude map (which emphasizes
low-gradient regions, whether blurred or naturally smooth) and $\gamma(\cdot)$ maps the
paired high-frequency statistic, min--max normalized over all released metadata (training and test, range $[-3830.9,-33.6]$), to a linearly decreasing scale in $[0.5,1.5]$; because the descriptor is target-derived and supplied at oracle inference, this fit is transductive. Only the prior encoder's first projection is
widened to accept the severity channel, with the added weights zero-initialized so
that the model is identical to the DAB at initialization (Section~\ref{sec:method});
the descriptor thus enters the prior encoder alongside the image and reaches the
restorer only through the compact prior it produces. At the feature level, feature-wise linear modulation instead maps
$\mvec$ through an embedding MLP $h_\phi$ to per-block affine parameters and
modulates the output of every transformer block across the generator,
\begin{equation}
(\mathbf{g}_\ell,\beta_\ell)=h_\phi(\mvec),\qquad
\mathrm{FiLM}(\mathbf{F}_\ell)=(1+\mathbf{g}_\ell)\odot\mathbf{F}_\ell+\beta_\ell,
\label{eq:film}
\end{equation}
following the standard FiLM formulation~\cite{film}. In the executed configuration,
$\mvec$ holds five metadata fields in this order: optical-flow motion magnitude,
paired high-frequency difference, blur window, blur duration, and noise estimate,
each z-score normalized with statistics fitted on the training split and without
imputation. The embedding is a three-layer MLP ($5\!\to\!128\!\to\!128\!\to\!128$,
SiLU). For each of the $38$ transformer blocks of the encoder, bottleneck, decoder,
and refinement stages, a linear layer maps the embedding to the $2C_\ell$ values
$(\mathbf{g}_\ell,\beta_\ell)$, which are broadcast over space and applied to that
block's output. These layers use Xavier-uniform weights with gain $0.01$ and zero
bias, so that $1+\mathbf{g}_\ell\!\approx\!1$ and $\beta_\ell\!\approx\!0$ at
initialization; with a zero metadata vector, the initialized model differs from the
DAB by a mean absolute $0.153$ intensity units on $[0,255]$. The embedding and
generators add $1.68$M parameters. A training-only auxiliary head global-average-pools
the $384$-channel bottleneck feature and predicts the normalized $\mvec$ through a
three-layer MLP ($384\!\to\!128\!\to\!128\!\to\!5$, final layer zero-initialized);
its mean $\ell_1$ loss is added with weight $0.05$, and its gradient reaches the restorer
through the bottleneck feature. A spatial variant follows spatial feature transform~\cite{sft}: in place of the metadata vector, each of the same $38$ blocks receives the severity map of \eqref{eq:sev}, bilinearly resized to the block's resolution and passed through a per-block three-layer convolutional network ($1\!\to\!32\!\to\!32\!\to\!2C_\ell$, $3\times3$ kernels, SiLU) whose output gives spatially varying $(\mathbf{g}_\ell,\beta_\ell)$, applied as in \eqref{eq:film}. Its final layer uses the same Xavier gain $0.01$ and zero bias; the variant adds $4.05$M parameters, is trained with the $\ell_1$ objective alone, and receives the oracle severity map at inference. At the output level, a tier-contrastive term~\cite{simclr,airnet} acts on a
two-layer MLP projection $\mathbf{z}(\hat\yb)$ of the prediction, with
$\ell_2$-normalized embeddings,
\begin{equation}
\mathcal{L}_{\text{con}}=-\frac{1}{|A|}\sum_{i\in A}\frac{1}{|P(i)|}\sum_{p\in P(i)}
\log\frac{e^{\,\mathbf{z}_i^{\top}\mathbf{z}_p/\tau}}{\sum_{k\neq i}e^{\,\mathbf{z}_i^{\top}\mathbf{z}_k/\tau}},
\label{eq:contrast}
\end{equation}
where $P(i)$ holds the other same-tier images in the mini-batch (the anchor
excluded), $A$ is the set of anchors with nonempty $P(i)$, and $\tau=0.1$. The arm
uses a physical batch of four images without accumulation, so each anchor is
contrasted against three others; with the training tier shares, about $74\%$ of
anchors have a same-tier positive, and the remainder are masked. The term is added
to $\mathcal{L}_1$ with weight $0.1$, and its gradients reach the restorer through
$\hat\yb$, so it shapes the reconstruction indirectly rather than supervising
pixels.

Finally, an unconditioned spectral-loss control uses no severity channel and adds a Fourier-magnitude reconstruction term with a fixed weight,
\begin{equation}
\mathcal{L}=\mathcal{L}_1+\lambda_f\big\lVert\,|\mathcal{F}(\yhat)|-|\mathcal{F}(\yb)|\,\big\rVert_1,
\label{eq:freq}
\end{equation}
where $\mathcal{F}$ is the orthonormally scaled two-dimensional real discrete
Fourier transform applied per colour channel, with the mean taken over all
coefficients and $|\cdot|$ is the elementwise magnitude, so the term
penalizes spectral-magnitude error only. The weight $\lambda_f=0.1$ is fixed and does
not depend on the descriptor; this arm therefore tests whether an additional spectral objective helps the plain three-channel baseline. It receives no descriptor and serves as an auxiliary-loss control, not as a descriptor-injection arm. It is also distinct from focal frequency
loss~\cite{focalfreq}, whose adaptive per-frequency weighting is not used here. The
frequency-domain error is penalized
in $\ell_1$ rather than $\ell_2$ for robustness to the heavy-tailed spectral
residuals typical of motion blur, where a few high-magnitude frequency components
would otherwise dominate an $\ell_2$ penalty; the term is thus a spectral-sparsity
prior on the residual rather than a Parseval-exact surrogate for PSNR. As
Section~\ref{sec:exp} reports, with the descriptor supplied at inference only the
input-level mechanism improves over the schedule-matched DAB: both feature-level variants, the output-level contrastive term, and the spectral-loss control fall below the baseline. Possible reasons are discussed
next.

\section{Proposed Method and Its Theory}
\label{sec:method}
The proposed method is the input-level mechanism, and its design follows a single
principle: a weak descriptor should be introduced so that the network
can use it where it is informative and ignore it elsewhere, without disturbing a
representation that already restores well on average. The severity map of
\eqref{eq:sev} is computed at full resolution. Sobel operators yield the horizontal
and vertical gradients whose smoothed magnitude forms $g(\xb)$; inverting and
normalizing produces a field that is high in low-gradient regions, and the global
factor $\gamma(m_{\text{isp}})$ rescales the entire field, decreasing linearly as the
high-frequency statistic grows. The map therefore combines a spatial term computed from the blurred input
with a global term from paired data, neither of which the DAB receives as an
explicit signal. Figure~\ref{fig:severity} visualizes the map across tiers, and
Figure~\ref{fig:pipeline} shows the complete pipeline, including the DiffIR backbone
into which the severity channel is integrated.

\begin{figure*}[t]
\centering
\includegraphics[width=\textwidth]{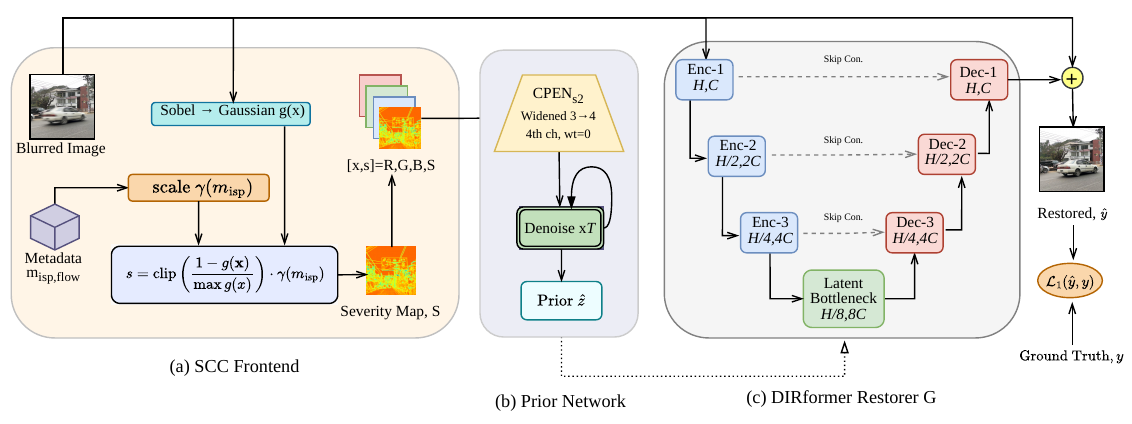}
\caption{\method{} integrated with DiffIR. \textbf{(a)}~A severity map $\svec$ is built from $\xb$ and $m_{\text{isp}}$ (Eq.~\eqref{eq:sev}) and concatenated to $\xb$; the prior encoder's first projection is widened with zero-initialized weights, $W'[:,48{:}64]=\mathbf{0}$, the only modified layer, so the network equals the DAB at initialization (Proposition~\ref{prop:identity}). \textbf{(b)}~CPEN$_{S2}$ encodes $[\xb;\svec]$ into a compact prior, refined by a DDPM reverse process in latent space to give $\hat{\mathbf{z}}$. \textbf{(c)}~The unchanged DIRformer restores $\xb$, conditioned at every block by $\hat{\mathbf{z}}$, with a global residual connection. The gradient term of $\svec$ uses $\xb$ alone; $m_{\text{isp}}$ requires paired data (Section~\ref{sec:prelim}).}
\label{fig:pipeline}
\end{figure*}

\begin{figure*}[t]
\centering
\includegraphics[width=\textwidth]{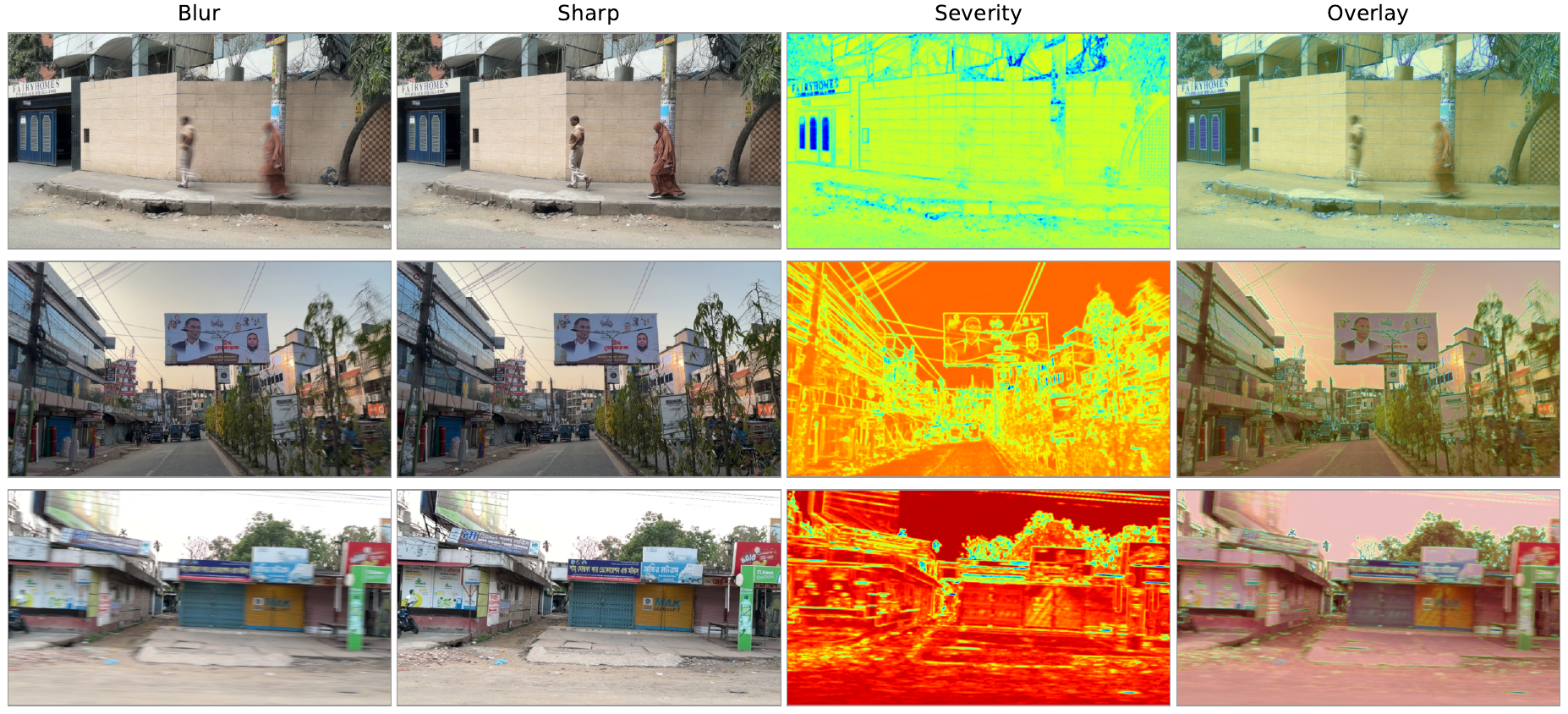}
\caption{Severity maps across tiers (rows: Easy, Medium, Hard). Columns: blurred input, ground truth, severity map $\svec$ (Eq.~\eqref{eq:sev}), and $\svec$ overlaid on the input. Each example is the test image closest to its tier's median high-frequency difference statistic. The map is high in low-gradient regions and scaled globally by $\gamma(m_{\text{isp}})$.}
\label{fig:severity}
\end{figure*}

The integration is minimal. The
prior encoder pixel-unshuffles its input by $4$, so its first projection has weight
$W\in\mathbb{R}^{C\times 48\times k\times k}$ for the image alone; the widened
weight $W'\in\mathbb{R}^{C\times 64\times k\times k}$ is formed with
$W'[:,0{:}48]=W$ and $W'[:,48{:}64]=\mathbf{0}$.

\begin{proposition}[Function-preserving initialization]
\label{prop:identity}
With $W'$ formed as above, the conditioned network at initialization produces an
output that is independent of $\svec$ and identical to the DAB, i.e.
$f_{\theta'}([\xb;\svec])=f_\theta(\xb)$ for every $\svec$. The fine-tuning
trajectory therefore begins at the DAB's training loss.
\end{proposition}
\begin{proof}
Pixel-unshuffle by $4$ maps each input channel to $16$ channels, so the three image
channels occupy unshuffled channels $0{:}48$ and the severity channel occupies
$48{:}64$. The projection output is $\sum_{c=0}^{63} W'[:,c]*u[c]$, where $u$ is the
unshuffled input. Because $W'[:,48{:}64]=\mathbf{0}$, the severity terms vanish,
leaving $\sum_{c=0}^{47}W[:,c]*u[c]$, which is exactly the DAB projection; every
downstream weight is inherited unchanged, so the two networks compute the same
output for all $\svec$. Equal outputs imply an equal loss at initialization.
\end{proof}

The statement is confined to initialization: it does not guarantee the
\emph{trained} model attains test error no worse than the DAB, since optimization
dynamics and generalization gap can move the endpoint either way. What it provides
is a common starting function for every mechanism compared, which is what the
controlled comparison requires.

The property is not unique to input concatenation, and the comparison does not rely
on the claim that it is: feature-wise modulation admits the same guarantee under the
reparameterization $\gamma_\ell=1+\delta\gamma_\ell$, $\beta_\ell=\delta\beta_\ell$
and is exact when the MLP's output layer is initialized to zero; the FiLM baseline
uses a near-zero initialization and therefore starts approximately at the pretrained
function (Section~\ref{sec:diag}). Both mechanisms thus start from, or near, the
pretrained function, and the feature level nonetheless does not exceed the baseline
(Section~\ref{sec:condstudy}). The experiments do not establish why. Possible factors include the form of the
control: FiLM applies one affine pair per channel across every block of an
already-converged generator, whereas \method{} alters only the input to the prior
encoder and leaves the restorer unchanged. For the remaining two arms, the contrastive term acts on the restored output only through an auxiliary projection, and the spectral term may contribute little where a converged $\ell_1$ model already retains little high-frequency error. These are candidate explanations, not tested
mechanisms.

Two further statements bear on these explanations. The first states when observing
the descriptor can lower the attainable risk; the second is a hypothesis about the
input and feature levels, tested below.

\begin{proposition}[Value of side information]
\label{prop:blind}
Let $R_{\xb}=\inf_{f}\mathbb{E}[\ell(f(\xb),\yb)]$ be the minimum risk over
predictors of $\xb$ alone, and $R_{\xb,\mvec}=\inf_{g}\mathbb{E}[\ell(g(\xb,\mvec),\yb)]$
the minimum risk over predictors that additionally observe the descriptor, and
assume both risks are finite and attained. Then $R_{\xb,\mvec}\le R_{\xb}$. Let
$r(a;\xb,\mvec)=\mathbb{E}[\ell(a,\yb)\mid\xb,\mvec]$,
$r^\star(\xb,\mvec)=\inf_a r(a;\xb,\mvec)$, and $f^\star$ an $\xb$-only minimizer.
Writing the conditional excess risk as
$\Delta(\xb,\mvec)=r(f^\star(\xb);\xb,\mvec)-r^\star(\xb,\mvec)\ge 0$,
we have $R_\xb-R_{\xb,\mvec}=\mathbb{E}[\Delta(\xb,\mvec)]$, so the inequality is
strict exactly when $\Delta>0$ with positive probability. A difference in
Bayes-optimal \emph{action} is not sufficient: under a loss with non-unique
minimizers two distinct actions can carry identical conditional risk. The
proposition concerns access to $\mvec$ at prediction time; it does not show that
training with $\mvec$ and then zeroing it at inference beats the $\xb$-only risk.
\end{proposition}
\begin{proof}[Proof sketch]
Every predictor $f(\xb)$ is also a predictor $g(\xb,\mvec)$ that ignores its second
argument, so the feasible set on the right contains the feasible set on the left and
the infimum cannot increase. By the tower property,
$R_\xb=\mathbb{E}[r(f^\star(\xb);\xb,\mvec)]$ and
$R_{\xb,\mvec}=\mathbb{E}[r^\star(\xb,\mvec)]$, so their difference is
$\mathbb{E}[\Delta]$, which is positive exactly when $\Delta>0$ on a set of positive
probability.
The full measure-theoretic derivation is given in Appendix~B of the supplementary material.
\end{proof}

Two consequences bound what the descriptor can be claimed to do. First, strictness
is not implied by the seven-decibel tier spread alone: a model observing $\xb$ may
already infer much of the severity from the degraded pixels themselves. Second,
whether the inequality is strict here is not settled by the experiments. The
test-time diagnostic of Section~\ref{sec:disc}, in which zeroing the descriptor
costs $0.02$~dB, shows that the trained network relies on $\mvec$ only slightly at
inference; it concerns one fitted model, not the optimal predictors compared in
the proposition.

One candidate explanation for the asymmetry between the input and feature levels
motivated the study. It is stated as a hypothesis, not a theorem, and is tested
below.

\medskip
\noindent\textbf{Hypothesis 1 (Gradient magnitude at the input level).}
\emph{At a converged $\ell_1$ optimum of the DAB, the gradient reaching the
zero-initialized input-channel weights, which scales with the agreement between
$\svec$ and the residual $\yb-\yhat$, is larger in relative terms than the gradient
reaching parameters that act multiplicatively on already-computed features.}

\medskip
Differentiating the loss through the zero-initialized weights $W'[:,48{:}64]$ produces, at
initialization, a term multiplying local patches of $\svec$ by the backpropagated
error signal carried through the full downstream Jacobian; the hypothesis is
therefore a statement about relative magnitudes, not an identity, and the
corresponding claim for feature-wise modulation is likewise not established
analytically. A normalized-step diagnostic at the shared DAB checkpoint
(Appendix~C) does not support it: one step changes the restored image by a mean absolute $0.039$ intensity units (on $[0,255]$) through the input channel and by $21.6$ through a probe of an earlier FiLM variant that modulates the compact prior, so under this parameterization the probed feature-level variant commands the larger change. Because the probe is not the reported per-block FiLM arm, it bears on the ordering only indirectly. The controls of Section~\ref{sec:controls} suggest
that structured severity inputs matter, but neither they nor the diagnostic
establish spatial localization of output corrections or a unique causal
mechanism.

Training fine-tunes all parameters from the DAB checkpoint (Algorithm~\ref{alg:scc}; settings in Section~\ref{sec:exp}). The only added inference cost is a single input channel, negligible beside
the backbone: widening the prior encoder's first convolution from $48$ to $64$ input
channels via a zero-initialized weight adds $9{,}216$ parameters, a $0.034\%$ increase
over the $26.9$M-parameter backbone ($26{,}911{,}934\!\rightarrow\!26{,}921{,}150$),
with correspondingly negligible change in floating-point operations and latency
(Section~\ref{sec:efficiency}).

\begin{algorithm}[t]
\caption{\method{} fine-tuning (one epoch)}
\label{alg:scc}
\begin{algorithmic}[1]
\Require DAB $f_\theta$; widened projection $W'$ with $W'[:,48{:}64]{=}\mathbf{0}$
\For{each mini-batch $(\xb,\yb)$}
  \State $\svec \gets$ severity map from $\xb$, $m_{\text{isp}}$ \Comment{Eq.~\eqref{eq:sev}}
  \State $\yhat \gets f_{\theta'}([\xb;\svec])$
  \State $\theta' \gets \theta' - \eta\,\nabla_{\theta'}\mathcal{L}_1(\yhat,\yb)$;\; update EMA
\EndFor

\end{algorithmic}
\end{algorithm}

\section{Experiments}
\label{sec:exp}
Evaluation is performed on the iPhoneBlur benchmark~\cite{iphoneblur}, comprising
$7{,}400$
blurred--sharp pairs from $51$ iPhone videos, stratified into Easy, Medium, and Hard
tiers and partitioned by source video into $36$ training and $15$ test captures,
giving $5{,}714$ training and $1{,}686$ test pairs with per-sample capture
metadata; formula, units, range, and provenance for every field are tabulated in
Appendix~A. Unless noted, the backbone is DiffIRS2~\cite{diffir}; all conditioned models and the schedule-matched baseline are fine-tuned from the same Domain-Adapted Baseline checkpoint, and
full-resolution test evaluation uses reflect-padding to a multiple of the network
stride. PSNR, SSIM~\cite{ssim}, and LPIPS~\cite{lpips} (AlexNet backbone) are reported
throughout.
Hyperparameters are held fixed across all mechanisms. The severity map of
\eqref{eq:sev} uses a $3\times3$ Sobel operator and $15\times15$ Gaussian smoothing
($\sigma=2.6$) of the gradient magnitude, normalized by its maximum with
$\epsilon=10^{-6}$; the scale is $\gamma=1.5-\tilde m$, with $\tilde m$ the min--max
normalized high-frequency statistic. The contrastive temperature is $\tau=0.1$ and the frequency weight $\lambda_f=0.1$. Fine-tuning uses AdamW ($\beta_1=0.9$, $\beta_2=0.99$, no weight
decay), a cosine schedule from $2\times10^{-5}$ (minimum $10^{-6}$) whose period is $34{,}284$ steps but which is stepped once per optimizer update ($17{,}136$ updates per arm), so the rate ends near $1.05\times10^{-5}$, and an EMA decay of
$0.999$. Full training configuration (epochs, batch/patch size, hardware, software,
seeds) is given in Appendix~G. Throughout, \emph{schedule-matched} means an identical number of
additional optimizer updates from the shared checkpoint, with identical schedule,
batch size, and augmentation, not matched floating-point operations, which differ
across mechanisms (Section~\ref{sec:efficiency}). Every model is evaluated at its
final epoch. The evaluation has
three parts: zero-shot cross-domain transfer, a fine-tuned comparison against six
baselines, and the conditioning study, followed by cross-backbone and per-image
significance analyses.

\subsection{Zero-Shot Transfer and the Domain Gap}
\label{sec:zeroshot}
Table~\ref{tab:zeroshot} quantifies the professional-to-consumer gap. A DiffIR model
trained on GoPro transfers to other professional blur datasets with the modest drops
typical of cross-evaluation, but on iPhoneBlur it falls to $27.22$~dB overall, a
cross-dataset gap consistent with the formation model of Section~\ref{sec:prelim},
although these experiments do not separate the consumer pipeline from other
differences between the datasets. This degradation motivates domain fine-tuning and, beyond it, the
descriptor-conditioning study that is the focus of this paper.

\begin{table}[t]
\centering
\caption{Zero-Shot Transfer of GoPro-Pretrained Models (PSNR, dB). Baseline values are from~\cite{iphoneblur}; DiffIR's GoPro and HIDE values are from~\cite{diffir} (RealBlur not reported) and its iPhoneBlur value is measured here. $^{\ast}$Evaluated here, not a new model.}
\label{tab:zeroshot}
\setlength{\tabcolsep}{2pt}
\small
\begin{tabular}{lccccc}
\toprule
Method & GoPro & HIDE & RealBlur-J & RealBlur-R & iPhoneBlur \\
\midrule
NAFNet~\cite{nafnet}       & $33.7$ & $31.3$ & $28.5$ & $35.1$ & $26.6$ \\
HINet~\cite{hinet}         & $32.9$ & $30.2$ & $27.2$ & $34.2$ & $26.8$ \\
Restormer~\cite{restormer} & $32.9$ & $31.2$ & $29.0$ & $35.4$ & $26.8$ \\
FFTformer~\cite{fftformer} & $34.0$ & $31.6$ & $29.2$ & $35.6$ & $26.9$ \\
\midrule
DiffIR~\cite{diffir}$^{\ast}$ & $33.20$ & $31.55$ & -- & -- & $27.22$ \\
\bottomrule
\end{tabular}
\end{table}

\subsection{Fine-Tuned Comparison}
\label{sec:finetune}
Fine-tuning on iPhoneBlur closes most of the gap. Table~\ref{tab:finetune} compares
the Domain-Adapted Baseline against six strong deblurring methods, all fine-tuned on
the target domain and evaluated at full resolution with stratified statistics. The
DAB reaches $31.23$~dB overall, competitive with or exceeding the strongest baselines, so the study of Section~\ref{sec:condstudy} starts from a strong base; the persistent Easy--Hard gap of roughly
seven decibels across every method reflects input difficulty, which an explicit
descriptor may help a model exploit.

\begin{table}[t]
\centering
\caption{Fine-Tuned Comparison on iPhoneBlur (PSNR, dB, per Tier; SSIM and LPIPS Overall). Baselines are published benchmark results~\cite{iphoneblur}, three-seed means, fine-tuned from GoPro weights. The DAB row is seed~1; its three-seed mean is $31.22$~dB (Section~\ref{sec:significance}).}
\label{tab:finetune}
\setlength{\tabcolsep}{2.5pt}
\footnotesize
\begin{tabular}{lcccccc}
\toprule
Method & Easy & Med. & Hard & Overall & SSIM & LPIPS \\
\midrule
MIMO-UNet~\cite{mimounet}    & $33.9$ & $29.7$ & $25.7$ & $29.3$ & $0.920$ & $0.075$ \\
Instruct-IR~\cite{instructir}& $34.5$ & $29.9$ & $25.4$ & $29.5$ & $0.924$ & $0.071$ \\
HINet~\cite{hinet}           & $34.6$ & $30.9$ & $27.1$ & $30.5$ & $0.938$ & $0.056$ \\
Restormer~\cite{restormer}   & $34.9$ & $31.2$ & $27.5$ & $30.9$ & $0.940$ & $0.053$ \\
FFTformer~\cite{fftformer}   & $35.0$ & $31.3$ & $27.7$ & $31.0$ & $0.943$ & $0.050$ \\
NAFNet~\cite{nafnet}         & $35.1$ & $31.6$ & $27.9$ & $31.2$ & $0.945$ & $0.049$ \\
\midrule
DAB (ours, DiffIR)           & $\mathbf{35.14}$ & $\mathbf{31.62}$ & $\mathbf{27.99}$ & $\mathbf{31.23}$ & $\mathbf{0.948}$ & $\mathbf{0.049}$ \\
\bottomrule
\end{tabular}
\end{table}

\subsection{Conditioning Study}
\label{sec:condstudy}
Table~\ref{tab:condstudy}, the central experiment, compares each configuration with the schedule-matched baseline. With the descriptor supplied at
inference, input-level severity-channel conditioning is the only configuration that
improves over the
schedule-matched baseline, on every tier ($+0.06$/$+0.08$/$+0.09$~dB on
Easy/Medium/Hard, $+0.07$~dB overall). Global FiLM modulation regresses substantially ($-0.32$~dB), its spatial SFT-style variant falls $0.03$~dB short, the output-level contrastive term regresses $0.12$~dB, and the spectral-loss control falls $0.02$~dB below baseline, within the range of the baseline's three seeds.
One possible reason (Section~\ref{sec:method}) is that the spectral term adds
gradient where an already-converged $\ell_1$ model retains little error. The full breakdown is in
Appendix~D. For every configuration that conditions the network, SSIM and LPIPS move with PSNR across the rows of
Table~\ref{tab:condstudy}, and the oracle input-level gain is consistent across all three metrics. The spectral-loss control is the exception: its PSNR is above both feature-level variants, yet its SSIM is the lowest in the table and its LPIPS is second only to global FiLM's. A plausible reason is that \eqref{eq:freq} constrains only the magnitude spectrum and leaves phase, which carries spatial structure, unpenalized. 

The feature-level arm was initialized near the baseline function under the specified parameterization (Section~\ref{sec:diag}). Its lower score in this run does not identify the cause of the difference: initialization, descriptor representation, auxiliary supervision, and optimization effects have not been independently isolated.

\begin{table*}[t]
\centering
\caption{Conditioning Study on iPhoneBlur (PSNR, dB).  SSIM and LPIPS are overall; $\Delta$ is the overall PSNR change. The descriptor is supplied at inference (oracle). The spectral-loss control receives no descriptor.}
\label{tab:condstudy}
\setlength{\tabcolsep}{6pt}
\small
\begin{tabular}{llccccccc}
\toprule
Level & Mechanism & Easy & Med. & Hard & Overall & $\Delta$ & SSIM & LPIPS \\
\midrule
--    & Schedule-matched DAB & $35.14$ & $31.62$ & $27.99$ & $31.23$ & -- & $0.948$ & $0.049$ \\
\midrule
Feat. & FiLM, global (near-identity~init) & $34.90$ & $31.28$ & $27.66$ & $30.91$ & $-0.32$ & $0.945$ & $0.053$ \\
Feat. & FiLM, spatial (SFT-style)    & $35.12$ & $31.59$ & $27.98$ & $31.20$ & $-0.03$ & $0.947$ & $0.050$ \\

Out.  & Contrastive & $35.04$ & $31.49$ & $27.87$ & $31.11$ & $-0.12$ & $0.946$ & $0.051$ \\
Ctrl. & Spectral-loss control (no descriptor) & $35.10$ & $31.59$ & $28.02$ & $31.21$ & $-0.02$ & $0.944$ & $0.052$ \\
\midrule
\rowcolor{green!8}
In.   & \textbf{\method{} (ours)} & $\mathbf{35.20}$ & $\mathbf{31.70}$ & $\mathbf{28.08}$ & $\mathbf{31.30}$ & $\mathbf{+0.07}$ & $\mathbf{0.949}$ & $\mathbf{0.048}$ \\
\bottomrule
\end{tabular}
\end{table*}

\subsection{Descriptor and Capacity Controls}
\label{sec:controls}
A modest gain invites the question of whether it comes from the descriptor's
\emph{information} or from a widened convolution's capacity.
Table~\ref{tab:controls} separates the two, holding architecture, recipe, and
budget fixed and corrupting only the fourth channel. The random-constant control
replaces the severity map with a per-image
random scalar broadcast to a uniform map: identical parameter count, no information
about the image. It recovers only $+0.01$~dB of \method{}'s $+0.07$~dB, suggesting that the gain tracks the descriptor's content rather than added capacity. Relative to the baseline's own seed variance (Section~\ref{sec:significance}), the random-constant run falls \emph{inside} the
range of DAB's three seeds, $[31.18, 31.26]$, whereas \method{}'s reported run lies
above it; single runs are indicative only.
In this run, added capacity alone did not reproduce the effect.

The shuffled-severity control addresses a finer question. Permuting the map's
pixels preserves its marginal distribution, and so the global severity level,
while destroying its spatial correspondence with the image. Shuffling collapses the gain to $-0.01$~dB, so \method{} exceeds its own shuffled control by $0.08$~dB. The global level alone, preserved by the shuffle, is therefore not what the
descriptor contributes. Because permutation also destroys spatial smoothness and
autocorrelation, this indicates that structured severity inputs matter; it does
not isolate residual alignment as the cause. 

The severity map of \eqref{eq:sev} itself multiplies an image-gradient term by a
global scale factor $\gamma$, and this construction admits the same style of ablation:
Table~\ref{tab:controls}
also reports the image-gradient term alone and the scale factor $\gamma$ alone (a spatially uniform map).
Neither component alone recovers the gain: the image-gradient term, which keeps the spatial
gradient pattern but drops the per-image global scale,
falls to $31.22$~dB ($-0.01$~dB); the scale factor $\gamma$ alone, which preserves the global
severity level but carries no spatial information, lands at $31.23$~dB ($0.00$~dB),
exactly on the baseline. All four corrupted or partial variants in
Table~\ref{tab:controls} lie within the range of the baseline's seeds; only the full product lies above that range. Each is a single fine-tuning run
corrupting a different aspect of the descriptor, so they are not replications and
do not substitute for repeated seeds. Together these single runs are consistent with a joint dependency: in this
configuration, neither the spatial term nor the global scale of \eqref{eq:sev}
reproduced the gain on its own.

\begin{table}[t]
\centering
\caption{Descriptor Ablations (Overall PSNR, dB). All runs share \method{}'s architecture, recipe, and budget and differ only in the fourth channel; single runs at seed~1. The baseline's three seeds span $[31.18,31.26]$.}
\label{tab:controls}
\setlength{\tabcolsep}{3pt}
\small
\begin{tabular}{llcc}
\toprule
Fourth-channel content & Information & PSNR & $\Delta$ \\
\midrule
None (schedule-matched DAB)  & --              & $31.23$ & -- \\
\midrule
Random constant              & none                & $31.24$ & $+0.01$ \\
Shuffled severity            & global (shuffled)   & $31.22$ & $-0.01$ \\
Image gradient only          & spatial only        & $31.22$ & $-0.01$ \\
Scale factor $\gamma$ only    & global (uniform)    & $31.23$ & $0.00$ \\
\rowcolor{green!8}
True severity (\method{})    & global + spatial    & $\mathbf{31.30}$ & $\mathbf{+0.07}$ \\
\bottomrule
\end{tabular}
\end{table}

\subsection{Cross-Backbone Generalization}
\label{sec:crossbackbone}
To test whether the finding is specific to a diffusion-prior backbone, the
input-level experiment is repeated on Restormer~\cite{restormer}, a transformer,
with the first convolution widened from three to four channels and the added
weights zero-initialized. For Restormer, the baseline and the conditioned model are
fine-tuned from the same GoPro-pretrained weights for an identical $20$ epochs, whereas the DiffIR comparison branches both arms from the converged DAB with
an identical additional budget. As a sanity check, the Restormer baseline
($30.95$~dB) reproduces the published fine-tuned result on this benchmark
($30.9$~dB, Table~\ref{tab:finetune}) to within $0.05$~dB.

Table~\ref{tab:crossbackbone} reports the result. Severity-channel conditioning
improves Restormer on every tier ($+0.09$~dB overall), in the same direction as on
DiffIR ($+0.07$~dB), so the input-level result is not specific to one backbone. A
conditioned NAFNet~\cite{nafnet} is also shown for context. Its baseline is the
published three-seed benchmark result~\cite{iphoneblur} rather than an independently
documented same-budget continuation in this study, so it is not used to establish a
matched improvement.

\begin{table}[t]
\centering
\caption{Cross-Backbone Comparison (PSNR, dB; Oracle Descriptor). Each \method{} model is compared with its backbone's baseline. The NAFNet baseline (Publ.) is the published three-seed benchmark result~\cite{iphoneblur}, shown at its published precision for context only; no matched difference is computed for it.}
\label{tab:crossbackbone}
\setlength{\tabcolsep}{3pt}
\small
\begin{tabular}{llccccc}
\toprule
Backbone & Model & Easy & Med. & Hard & Overall & $\Delta$ \\
\midrule
\multirow{2}{*}{DiffIR (diffusion)}
 & DAB          & $35.14$ & $31.62$ & $27.99$ & $31.23$ & -- \\
 & \method{}    & $35.20$ & $31.70$ & $28.08$ & $31.30$ & $+0.07$ \\
\midrule
\multirow{2}{*}{NAFNet (conv.)}
 & Publ.        & $35.1$  & $31.6$  & $27.9$  & $31.2$  & -- \\
 & \method{}    & $35.26$ & $31.72$ & $28.03$ & $31.31$ & -- \\
\midrule
\multirow{2}{*}{Restormer (transf.)}
 & DAB          & $34.92$ & $31.32$ & $27.72$ & $30.95$ & -- \\
 & \method{}    & $35.00$ & $31.42$ & $27.80$ & $31.04$ & $+0.09$ \\
\bottomrule
\end{tabular}
\end{table}

\subsection{Per-Image Significance}
\label{sec:significance}
Aggregate PSNR can conceal whether a small gain is broad or driven by a few
images. Table~\ref{tab:wilcoxon} reports a paired per-image analysis of \method{}
against the schedule-matched DAB on DiffIR: for each of the $1{,}686$ test images the
PSNR difference is taken and tested with a two-sided Wilcoxon signed-rank
test~\cite{wilcoxon}. \method{} wins on $87.2\%$ of
images overall and on more than $86\%$ within every tier, and the rank-biserial
effect size~\cite{kerby2014simple} of $r\approx0.72$--$0.76$ indicates that the
effect is consistent in direction rather than large in magnitude. The bootstrap
confidence interval~\cite{efron1979bootstrap} on the overall mean,
$[0.067,0.073]$~dB, excludes zero while remaining consistent with the modest magnitude reported throughout.

Test frames are drawn from
$15$ held-out source videos, so frames within a video are not mutually
independent, and the frame-level significance above is therefore optimistic about
generalization to \emph{new videos}. The correct unit for that question is the
video: aggregating to per-video means gives a Wilcoxon $p=0.0002$, with $14$
videos favouring \method{} and one near-zero ($-0.003$~dB); a percentile bootstrap
over video means gives an equal-video $95\%$ CI of $[0.044,0.084]$~dB for the
mean of per-video means ($+0.065$~dB; the frame-weighted mean is $+0.072$~dB), and an exact sign-flip
permutation test agrees ($p=0.0002$). This describes between-video consistency for
the evaluated model pair; it does not remove training-run uncertainty. Excluding the smallest video (two frames)
changes this negligibly ($p=0.0004$; CI $[0.041,0.082]$~dB). Figure~\ref{fig:violin}
shows this distribution; per-video means and frame counts are in Appendix~F. The
paired design compares the two arms on identical test images, which removes
image-to-image variance but not variance between training runs.

Table~\ref{tab:multiseed} reports three training seeds per arm. Over
three seeds per arm the ranges do not overlap and the gain of means is
$+0.10$~dB, larger than the $+0.07$~dB single-run figure; the reported run is seed~1 of
each arm ($31.23$ and $31.30$~dB), and the difference reflects seed-to-seed
variation. Seed variance is smaller than the benchmark's reported figure since
all runs fine-tune from a fixed checkpoint rather than from scratch. With three runs
per arm a paired test over the three seed pairs ($t=3.04$, $\mathrm{df}=2$, $p=0.093$, 95\% CI $[-0.040,+0.234]$) has
little power, so this is reported as corroboration, not a primary claim. 

\begin{table}[t]
\centering
\caption{Per-Image Significance of \method{} Over the Schedule-Matched DAB on
iPhoneBlur (DiffIR backbone). $p$-value is a two-sided Wilcoxon signed-rank test;
$r$ is the rank-biserial correlation. Oracle descriptor at inference; frame-level, see Section~\ref{sec:significance} for the video-clustered check.}
\label{tab:wilcoxon}
\setlength{\tabcolsep}{2.5pt}
\scriptsize
\begin{tabular}{lccccccc}
\toprule
Tier & $N$ & Med.~$\Delta$ (dB) & Win \% & $p$-value & 95\% CI (dB) & $r$ \\
\midrule
Easy    & $280$  & $+0.051$ & $86.8$ & $<10^{-33}$  & -- & $0.74$ \\
Medium  & $953$  & $+0.079$ & $87.9$ & $<10^{-108}$ & -- & $0.76$ \\
Hard    & $453$  & $+0.093$ & $86.1$ & $<10^{-50}$  & -- & $0.72$ \\
\midrule
\textbf{Overall} & $\mathbf{1686}$ & $\mathbf{+0.079}$ & $\mathbf{87.2}$ & $\mathbf{<10^{-188}}$ & $\mathbf{[0.067,0.073]}$ & $\mathbf{0.74}$ \\
\bottomrule
\end{tabular}
\end{table}

\begin{table}[t]
\centering
\caption{Multi-Seed Comparison on iPhoneBlur (Overall; Oracle Descriptor). Three seeds per arm; only data order and augmentation differ. Paired test over seeds: $p=0.093$ (Section~\ref{sec:significance}).}
\label{tab:multiseed}
\setlength{\tabcolsep}{4pt}
\small
\setlength{\tabcolsep}{3pt}
\footnotesize
\begin{tabular}{llccc}
\toprule
Model & Seed & PSNR & SSIM & LPIPS \\
\midrule
\multirow{4}{*}{DAB}
 & 1 & $31.23$ & $0.948$ & $0.049$ \\
 & 2 & $31.18$ & $0.947$ & $0.050$ \\
 & 3 & $31.26$ & $0.949$ & $0.048$ \\
 & \textbf{mean} & $31.22{\scriptstyle\,\pm 0.04}$ & $0.948{\scriptstyle\,\pm 0.001}$ & $0.049{\scriptstyle\,\pm 0.001}$ \\
\midrule
\multirow{4}{*}{\method{}}
 & 1 & $31.30$ & $0.949$ & $0.048$ \\
 & 2 & $31.34$ & $0.952$ & $0.050$ \\
 & 3 & $31.32$ & $0.950$ & $0.048$ \\
\rowcolor{green!8}
 & \textbf{mean} & $\mathbf{31.32}{\scriptstyle\,\pm 0.02}$ & $\mathbf{0.950}{\scriptstyle\,\pm 0.002}$ & $\mathbf{0.049}{\scriptstyle\,\pm 0.001}$ \\
\bottomrule
\end{tabular}
\end{table}

\begin{figure}[t]
\centering
\includegraphics[width=\columnwidth]{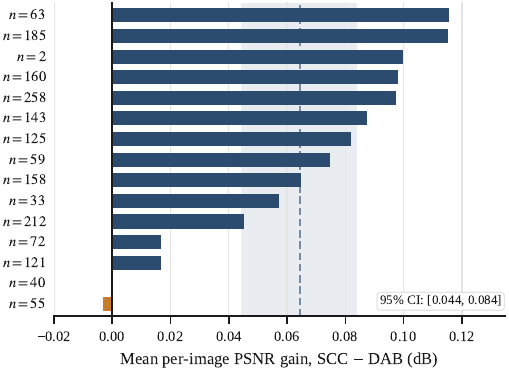}
\caption{Mean PSNR gain ($\method - \text{DAB}$, oracle descriptor) per held-out video, sorted by effect size; $n$ is the number of frames per video. The shaded band is the $95\%$ equal-video bootstrap CI. Fourteen of $15$ videos favour \method{}.}
\label{fig:violin}
\end{figure}

\subsection{Efficiency}
\label{sec:efficiency}
Table~\ref{tab:eff} places the input-level mechanism's cost in context against the
alternatives, across the backbones of Table~\ref{tab:crossbackbone}. The severity
channel widens a single convolution, adding a small fraction of a percent of
parameters ($9{,}216$, $0.034\%$, on DiffIR; $576$ on NAFNet; $432$ on
Restormer; the per-backbone differences follow from where each reads the image,
detailed in Appendix~G). Latency is unaffected on every backbone, within
measurement noise. FiLM requires $1.68$M additional parameters,
$182\times$ \method{}'s overhead, and still regresses $0.32$~dB below the baseline.
Making the control spatial rather than global costs more, not less: the SFT-style
variant adds $4.05$M parameters, roughly $440\times$ \method{}'s, and increases
FLOPs by $45.7\%$ and measured latency by a closely corresponding $43.4\%$, since
spatial modulation applies a real convolution at every pixel of every conditioned
block; this is the only mechanism in Table~\ref{tab:eff} whose latency tracks its
compute increase, and it still falls short of the baseline. The cheapest configuration is also the only one that improves over the baseline; added capacity does not account for the gain.

\begin{table}[t]
\centering
\caption{Inference Cost at $720\times1280$ on One GPU. The contrastive head and the spectral loss are used only in training, so those rows equal the DAB. The lower block covers the other two backbones (Section~\ref{sec:crossbackbone}).}
\label{tab:eff}
\setlength{\tabcolsep}{3pt}
\footnotesize
\begin{tabular}{lcccc}
\toprule
Mechanism & Params (M) & Added (M) & FLOPs (G) & Lat. (ms) \\
\midrule
\multicolumn{5}{l}{\textit{DiffIR backbone}} \\
Schedule-matched DAB  & $26.912$ & --      & $3380.5$ & $3397$ \\
\midrule
FiLM, global & $28.593$ & $1.681$  & $3380.5$ & $3454$ \\
FiLM, spatial (SFT)  & $30.962$ & $4.05$   & $4926.5$       & $4870$   \\
Contrastive          & $26.912$ & $0$      & $3380.5$ & $3397$ \\
Spectral-loss control & $26.912$ & $0$      & $3380.5$ & $3397$ \\
\rowcolor{green!8}
\method{} (ours)     & $26.921$ & $0.0092$ & $3384.7$ & $3381$ \\
\midrule
\multicolumn{5}{l}{\textit{Cross-backbone (Section~\ref{sec:crossbackbone})}} \\
NAFNet (publ.)         & $67.889$ & --      & $1780.8$ & $1179$ \\
\rowcolor{green!8}
\method{} (NAFNet)   & $67.889$ & $0.0006$ & $1781.9$ & $1190$ \\
DAB (Restormer)      & $26.127$ & --      & $3965.4$ & $4351$ \\
\rowcolor{green!8}
\method{} (Restormer)& $26.127$ & $0.0004$ & $3966.2$ & $4304$ \\
\bottomrule
\end{tabular}
\end{table}

\subsection{Qualitative Results}
\label{sec:qualitative}
Figure~\ref{fig:quali} compares the blurred input, the \method{} restoration, and the
ground truth across the difficulty spectrum, with two magnified regions per example.
\method{} recovers edges, text, and repeated texture, most clearly on Easy and Medium inputs.
The figure does not include the DAB column, so it illustrates restoration quality
rather than the $+0.07$~dB margin over the baseline, which is quantified in
Table~\ref{tab:condstudy}; a DAB comparison and failure cases are in Appendix~H.

\begin{figure*}[t]
\centering
\includegraphics[width=\textwidth]{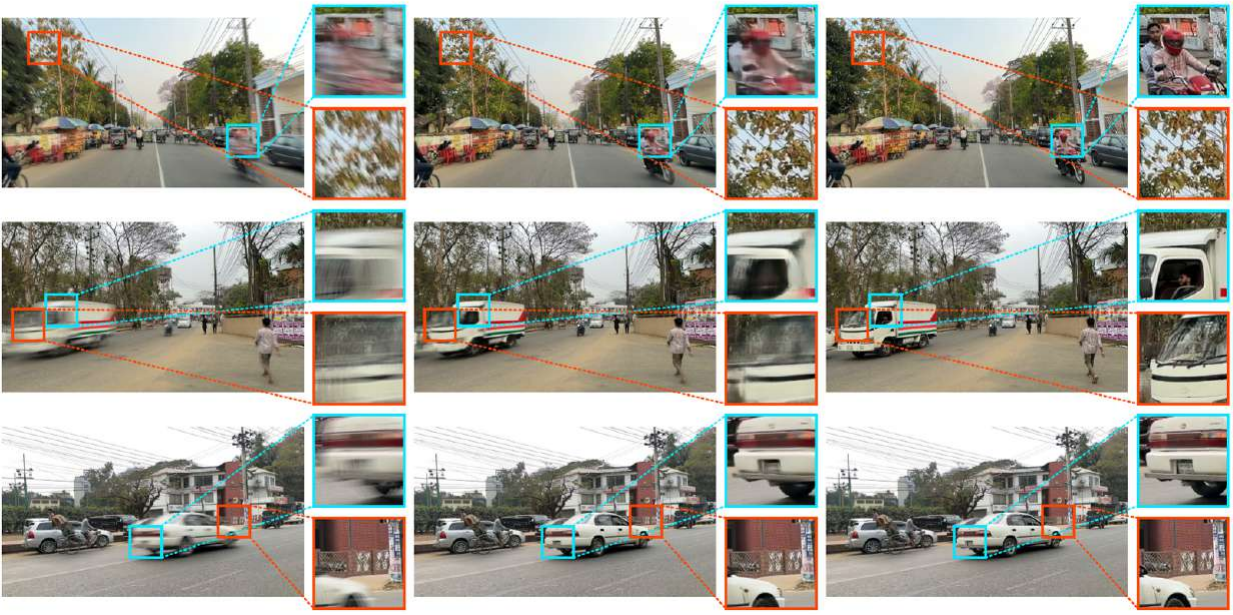}
\caption{Qualitative comparison on iPhoneBlur. Columns: blurred input, \method{}
restoration, ground truth; two magnified crops accompany each example (orange and
cyan boxes). Rows span the difficulty spectrum. }
\label{fig:quali}
\end{figure*}

\section{Discussion}
\label{sec:disc}
Three properties characterize the oracle input-level gain: the sign
is uniform across tiers; the direction holds on $87\%$ of images
(Table~\ref{tab:wilcoxon}) and on both matched backbones; and it tracks the
descriptor's content rather than added capacity, since a random constant recovers
only $+0.01$~dB.

The gain must be read alongside an exploratory inference ablation. Feeding zeros through the
fourth channel instead of the true severity map lowers overall PSNR by $0.02$~dB
($31.30\!\rightarrow\!31.28$), with SSIM falling from $0.949$ to $0.944$ and LPIPS
unchanged at $0.048$. The zeroed configuration sits $0.05$~dB above the baseline's single-run value,
so in this run most of the oracle gain survives without the descriptor, at an SSIM cost of
$0.004$ relative to the baseline. It uses one trained checkpoint, so it is an exploratory inference-ablation result, and no conclusion about deployment performance across seeds or videos is drawn from it. One explanation, which these experiments do not test directly, is that the channel
acts mainly during \emph{training}, so that the trained weights carry most of the
benefit without the descriptor at inference. Removing the fourth channel entirely, reverting to a three-channel first
convolution, gives the identical result. That equality is not incidental: both
configurations contribute $W_{\text{sev}}\!\cdot\!\mathbf{0}$ at the first
convolution and are algebraically the same computation. Under the
determinism-pinned protocol of Appendix~E the two agree exactly, to every reported
decimal and to zero maximum pixel difference. In this run, the trained network relies on the descriptor only slightly at inference. The severity map is used here as \emph{privileged information}~\cite{vapnik2009lupi}: a signal available during fine-tuning, derived in part from the sharp target, and unavailable at deployment, in the spirit of generalized
distillation~\cite{lopezpaz2016distillation} and its application to
super-resolution~\cite{lee2020pisr}.

The largest per-image regressions, the negative tail of
Figure~\ref{fig:violin}, share an interpretable structure: blur confined to a
small foreground region against a sharp background, such as a fast-moving
subject crossing a static scene. Because $\gamma(m_{\text{isp}})$ in \eqref{eq:sev} is a single scalar per image, it does not distinguish concentrated from diffuse blur, and on such scenes the descriptor can
be miscalibrated. This is likely a limitation of the descriptor rather than of input-level conditioning, and it suggests descriptors encoding the spatial concentration of
degradation, not only its average magnitude. More broadly, the study evaluates few configurations per level, and global FiLM is a global modulation: one affine pair per channel, uniform across space, which leaves open
whether a spatially varying feature-level mechanism might close the gap to the input
level. The spatial feature-transform row of Table~\ref{tab:condstudy} addresses this:
an SFT-style modulation driven by the same severity map, under the identical
protocol, closes most of global FiLM's gap ($-0.32\!\rightarrow\!-0.03$~dB) but
still does not exceed the baseline, and its margin is smaller than the baseline's own seed spread. A spatial control therefore recovers most of what a global one discards without
closing the gap; both feature-level variants act on already-computed features
rather than on the input. This broadens the observation from one instantiation to two, though it remains a
statement about the variants evaluated
here rather than about feature-level conditioning in general; learned modulation,
alternative contrastive constructions, and richer descriptors remain untested.
The matched cross-backbone evidence covers diffusion and transformer backbones (Table~\ref{tab:crossbackbone}), but the descriptor itself is validated on a
single benchmark, since iPhoneBlur is, to the authors' knowledge, the only paired
consumer-device
deblurring dataset that ships the capture metadata the descriptor requires; whether
the finding transfers to other consumer pipelines remains untested. 

The Sobel-gradient term of
\eqref{eq:sev} needs only the blurred image; the scalar $m_{\text{isp}}$ is
computed from paired data (Section~\ref{sec:prelim}) and is unavailable in deployment.
The descriptor as defined is therefore a training-time signal. A deployable model would zero the channel; this paper reports only an exploratory single-run ablation of that setting, and its evaluation across seeds and videos is left to future work.
Establishing an input-only descriptor, whether by a blind high-frequency
estimator or by a vendor statistic where one is documented, is the natural route to
per-image adaptation at test time and is not claimed here. Code, configs, and
per-image predictions are described in Appendix~I.

\section{Conclusion}
\label{sec:concl}

This work framed consumer-device motion deblurring as a domain-adaptation problem and asked, for a measured degradation descriptor, at which level of the network it should enter. Descriptor-conditioned deblurring is established; what has not been reported for motion deblurring is a matched-budget comparison of input-, feature-, and output-level conditioning. Under that protocol and with the descriptor supplied at inference, input-level severity-channel conditioning is the only evaluated configuration that exceeds the schedule-matched baseline, consistently across difficulty tiers, while the evaluated feature-level variants, including a spatially varying control that costs more than its global counterpart, the output-level contrastive term, and an unconditioned spectral-loss control do not. The gain is small: per-video means give $p=0.0002$, with 14 of 15 videos favouring SCC for the evaluated model pair, but three training seeds do not establish robustness to retraining, and single-run controls only suggest that the effect tracks the descriptor's content rather than added parameters. The descriptor is partly derived from the paired sharp target and is therefore unavailable in deployment; in one exploratory run, zeroing it at inference cost $0.02$~dB of PSNR and lowered SSIM slightly below the baseline. Evaluating that deployable setting across seeds and videos, developing a descriptor computable from the blurred input alone, and extending evaluation beyond one consumer benchmark are natural directions for future work.

The code and per-image predictions are publicly available at
\url{https://github.com/C-loud-Nine/Severity-Channel-Conditioning}.

\section*{Conflict of Interest}
The authors declare that they have no conflicts of interest relevant to this work.

\bibliographystyle{IEEEtran}
\bibliography{main}

\clearpage
\twocolumn[{\centering{\LARGE Supplementary Material}\par\vspace{1.5em}}]

\noindent This supplementary material accompanies the main paper above; equations, tables, figures, and propositions are numbered
consecutively with the main paper. All cross-references of the form
``Appendix~X'' in the main text refer to the corresponding section below.

\section*{Appendix A: Metadata Descriptor Provenance}
\label{app:provenance}
Table~\ref{tab:provenance} gives the exact formula, units, observed range on the
$1{,}686$-image test set, and provenance label for the two terms of the severity map and for the motion and
blur-window fields introduced in Section~III of the main paper; the remaining two
metadata fields are described below the table. The
provenance vocabulary is: \textbf{x-only} (computable from the blurred input
alone, hence available at deployment), \textbf{y-derived} (requires the paired
sharp ground truth, hence training-time only), \textbf{adjacent-frame} (requires
neighbouring video frames beyond the single blurred frame), \textbf{EXIF/ISP}
(exposed by the capture pipeline's own metadata), and \textbf{synthesis-only}
(a property of how the benchmark's blur pairs were synthesized, with no
device-side analogue).

\begin{table*}[t]
\centering
\caption{Metadata Descriptor Field Provenance}
\label{tab:provenance}
\small
\begin{tabular}{p{3.4cm}p{8.0cm}p{2.6cm}p{2.4cm}}
\toprule
Field & Formula / Definition & Units & Provenance \\
\midrule
Paired high-frequency difference ($m_{\text{isp}}$, key \texttt{isp\_diff}) &
high-frequency energy (mean absolute Laplacian of luminance, as defined by the benchmark~\cite{iphoneblur}) of the blurred observation $\xb$ minus that of the sharp image $\yb$, min--max normalized over all released metadata, training and test (range $[-3830.9,-33.6]$; transductive) & difference of Laplacian energies & y-derived \\
Image gradient term $g(\xb)$ & Gaussian-smoothed Sobel gradient-magnitude map of
$\xb$ (Eq.~(2) of the main text); high in low-gradient regions, blurred or naturally smooth &
dimensionless (gradient magnitude) & x-only \\
Optical-flow motion magnitude (\texttt{motion}) & Mean Farneb\"ack dense optical
flow magnitude averaged over the synthesis window: $m_i = \frac{1}{W_i}
\sum_{t=t_s}^{t_e} \mathbb{E}_{(x,y)}[\lVert v_t(x,y)\rVert]$, where $v_t$ is
dense flow between consecutive source frames~\cite{iphoneblur,farneback} &
px/frame (range $0.0$--$28.1$) & adjacent-frame \\
Temporal blur window (\texttt{blur\_window}) & Adaptive number of consecutive
high-frame-rate source frames averaged to synthesize each blurred
sample~\cite{iphoneblur} & frames (range $3$--$21$) & synthesis-only \\
\bottomrule
\end{tabular}
\end{table*}

\noindent \textbf{Scope of specification.} The two terms \method{} itself uses,
$m_{\text{isp}}$ (the benchmark's \texttt{isp\_diff}) and the image-gradient term
$g(\xb)$, are fully specified above by Eq.~(2) of the main text, and the
optical-flow magnitude and temporal blur window are specified from the
benchmark's published metadata schema~\cite{iphoneblur}. The remaining two
quantities named in Section~III are released as the metadata columns
\texttt{blur\_duration\_ms} and \texttt{noise\_estimate}. The former is the temporal extent of the synthesis window, the window length (3--21 frames) divided by the source video's frame rate, in milliseconds ($12.5$--$120.7$~ms); the latter is a per-sample sensor-noise estimate included in the benchmark's metadata files ($1.71$--$8.38$); the schema and sample records used are supplied with the review materials. No field has missing values. The FiLM arm z-scores its five inputs with training-split statistics ($n=5{,}714$; mean / SD): motion $11.081$ / $5.561$; paired high-frequency difference $-1399.748$ / $673.741$; blur window $7.217$ / $6.125$; blur duration $36.753$ / $32.318$; noise estimate $4.748$ / $1.123$. The \method{} scale instead min--max normalizes the paired high-frequency difference over all released metadata, range $[-3830.9,-33.6]$. No result reported for \method{} depends on them.

\medskip
\noindent The provenance column matters for the deployment discussion of Sections~III and VII. A field marked synthesis-only is available at benchmark-construction
time but not at capture time, and a field derived from $\yb$ is unavailable at
inference by definition. Of the terms \method{} uses, the image gradient is
computable from the blurred input alone, while $m_{\text{isp}}$ is
$\yb$-derived; Section~VII of the main text evaluates the setting in which the whole channel is zeroed at inference.

\section*{Appendix B: Full Proof of Proposition 2 (Value of Side Information)}
\label{app:prop2proof}
This appendix gives the complete derivation of Proposition 2 (main text,
Section~V), whose main-text proof is an abbreviated sketch.

\begin{proof}[Full proof]
Let $\mathcal{F}$ denote the set of measurable functions $f:\mathcal{X}\to
\mathcal{Y}$ and $\mathcal{G}$ the set of measurable functions
$g:\mathcal{X}\times\mathcal{M}\to\mathcal{Y}$. Every $f\in\mathcal{F}$ induces a
$g\in\mathcal{G}$ via $g(\xb,\mvec) := f(\xb)$, so $\mathcal{F}\subseteq
\mathcal{G}$ under this embedding, and consequently
\[
R_{\xb,\mvec} = \inf_{g\in\mathcal{G}} \mathbb{E}[\ell(g(\xb,\mvec),\yb)]
\;\le\; \inf_{f\in\mathcal{F}} \mathbb{E}[\ell(f(\xb),\yb)] = R_{\xb}.
\]
For the gap, assume both risks are finite and attained, and define the conditional
risk $r(a;\xb,\mvec)=\mathbb{E}[\ell(a,\yb)\mid\xb,\mvec]$ and its infimum
$r^\star(\xb,\mvec)=\inf_a r(a;\xb,\mvec)$, attained by a measurable selection
$g^\star$. Let $f^\star$ be any minimizer of the $\xb$-only risk. By the tower
property, $R_{\xb}=\mathbb{E}[r(f^\star(\xb);\xb,\mvec)]$ and
$R_{\xb,\mvec}=\mathbb{E}[r^\star(\xb,\mvec)]$. Since $f^\star(\xb)$ is a feasible
action for every $(\xb,\mvec)$, the conditional excess risk
$\Delta(\xb,\mvec)=r(f^\star(\xb);\xb,\mvec)-r^\star(\xb,\mvec)$ is nonnegative, and
\[
R_{\xb}-R_{\xb,\mvec}=\mathbb{E}[\Delta(\xb,\mvec)]\;\ge\;0 .
\]
The inequality is strict exactly when $\Delta>0$ on a set of positive probability.
This condition concerns risks, not actions: when the conditional minimizer is not
unique, as is possible under $\ell_1$ loss, $g^\star$ may differ from $f^\star$ while
$\Delta=0$. For example, with $\yb\in\{0,1\}$ equiprobable and $\xb,\mvec$
uninformative, both $0$ and $1$ minimize absolute loss with risk $0.5$. The
proposition concerns access to $\mvec$ at prediction time and makes no claim about
a model trained with $\mvec$ and evaluated without it.
\end{proof}

\section*{Appendix C: Normalized-Step Diagnostic (Hypothesis 1)}
\label{app:hyp1}
This appendix reports the measurements referenced in Section~V (Hypothesis~1) of
the main text, comparing the input-channel mechanism (\method) against the
feature-level mechanism (FiLM) at the shared Domain-Adapted Baseline checkpoint.

\textbf{Protocol.} Both arms are measured on the same $64$ training crops of
$256\times256$ px, with identical crop coordinates, the same $\ell_1$ objective,
and all other weights held fixed. For \method{} the measured parameters are the
zero-initialized severity columns of the widened prior-encoder projection
($9{,}216$ weights); for FiLM they are the eight tensors ($18{,}112$ parameters) of an earlier FiLM variant that modulates DiffIR's compact prior from four metadata fields, attached to the same DAB checkpoint and reset to identity ($\gamma$ bias $1$, $\beta$ bias $0$) so that both mechanisms are
measured in their initialization state. The gradient norm is normalized by
parameter count; the signal-to-noise ratio is the mean gradient magnitude over
its standard deviation across the $64$ samples; the functional output change is
the mean absolute pixel-space change in $\yhat$ after one step of size
$10^{-3}$ along the normalized gradient direction.

\begin{table*}[t]
\centering
\caption{Normalized-Step Diagnostic: Input Channel vs.\ FiLM at the Shared DAB Checkpoint}
\label{tab:hyp1}
\small
\begin{tabular}{lccc}
\toprule
Mechanism & Grad.\ norm (normalized)$^{\dagger}$ & Grad.\ SNR & Output change (intensity, $[0,255]$) \\
\midrule
Input (\method) & $1.25\times10^{-5}$ & $0.128$ & $0.039$ \\
Feature (FiLM)  & $2.35\times10^{-2}$ & $0.131$ & $21.637$ \\
\bottomrule
\end{tabular}

\vspace{2pt}
\parbox{0.9\textwidth}{\footnotesize $^{\dagger}$The gradient-norm column is not
comparable across the two rows and is reported for completeness only. The
\method{} parameter is a convolution weight multiplying a severity input in
$[0,1]$, whereas FiLM's are biases on a multiplier applied to feature
activations, so their raw magnitudes differ by the activation scale rather than
by any property of the mechanisms. The output-change column depends on the chosen
parameterization and is a sensitivity, not a ranking.}
\end{table*}

\noindent \textbf{Result.} One normalized gradient step changes the restored image
by a mean absolute $0.039$ intensity units (on $[0,255]$) through the input
channel and by $21.637$ through the probed FiLM tensors, roughly $550$ times larger for FiLM. The
gradient signal-to-noise ratios are near-identical ($0.128$ against $0.131$), so
neither estimate is the noisier of the two.

\medskip
\noindent \textbf{Interpretation.} The original expectation, that the input level
prevails because it receives the larger gradient, is not supported: under this diagnostic the probed FiLM variant produces the larger output change. Three limits apply. First, a unit step in two differently scaled
parameter spaces is not an equal functional perturbation, so the ratio is a
sensitivity under the chosen parameterization, not a scale-free ranking; a single
step also does not reproduce AdamW training dynamics. Second, the probe uses that earlier prior-modulating FiLM variant ($18{,}112$ parameters, reset to identity), not the per-block FiLM arm of Table~III ($1.68$M parameters), so it bears on the gradient-magnitude hypothesis only indirectly. Third, the diagnostic measures magnitude, not where changes occur. The
controls of Section~VI-D, a shuffled map and a random constant, indicate that
structured severity inputs matter, but neither they nor this diagnostic establish
spatial localization of output corrections or a unique causal mechanism.

\section*{Appendix D: Spectral-Loss Control (No Descriptor): Full Breakdown}
\label{app:objarm}
This appendix reports the complete per-tier and per-metric results for the
unconditioned spectral-loss control referenced in
Section~VI-C (Table~III) of the main text.

The arm adds only a fixed-weight ($\lambda_f=0.1$) spectral loss to the
unmodified three-channel baseline. An earlier configuration that also retained the
severity channel ($31.19$~dB overall) combined two mechanisms and is excluded.

\begin{table}[h]
\centering
\caption{Spectral-Loss Control (No Descriptor), Full Breakdown (PSNR dB / SSIM / LPIPS)}
\label{tab:objfull}
\footnotesize
\begin{tabular}{lccc}
\toprule
Tier & PSNR & SSIM & LPIPS \\
\midrule
Easy    & $35.100$ & $0.9710$ & $0.0310$ \\
Medium  & $31.589$ & $0.9455$ & $0.0500$ \\
Hard    & $28.020$ & $0.9253$ & $0.0703$ \\
\midrule
Overall & $31.213$ & $0.9443$ & $0.0523$ \\
\bottomrule
\end{tabular}
\end{table}

\noindent \textbf{Result.} $\Delta$ vs.\ the schedule-matched baseline
($31.23$~dB): $-0.02$~dB. The arm lies within the range of the baseline's three
seeds, $[31.18, 31.26]$; this is a descriptive comparison, not an equivalence test. The weighted tier
mean reproduces the reported overall figure exactly.

\medskip
\noindent The per-tier breakdown adds detail the aggregate row cannot show. The
deviation from the schedule-matched baseline is small on every tier and changes
sign across them: $-0.04$~dB on Easy, $-0.03$~dB on Medium, and $+0.03$~dB on
Hard. No tier deviates by more than $0.04$~dB, and the signs do not agree, so the
arm is not uniformly harmful in the way the feature-level mechanism is
(Table~III).

\section*{Appendix E: Zero-vs-Drop Determinism Audit}
\label{app:zerodrop}
This appendix reports the full protocol and results for the determinism-pinned
zero-vs-drop audit referenced in Section~VII (Discussion) of the main text.

\textbf{Protocol.} Both conditions use the identical trained \method{}
checkpoint. \emph{Zeroed}: the severity channel is set to zero at inference.
\emph{Dropped}: the weight slice $W'[:,48{:}64]$ of the first projection, which
receives the sixteen unshuffled severity channels, is masked to zero (equivalent to
reverting to the 48-channel DAB projection). For both conditions:
\texttt{torch.use\_deterministic\_algorithms(True)}, \texttt{cudnn.deterministic
= True}, \texttt{cudnn.benchmark = False}; a single fixed \texttt{manual\_seed =
0} is set immediately before each forward pass, and the diffusion
prior's initial latent $\hat{Z}_T$ is drawn once and reused identically across
both conditions.

\begin{table*}[t]
\centering
\caption{Zero-vs-Drop Determinism Audit}
\label{tab:zerodrop}
\small
\begin{tabular}{lccc}
\toprule
Condition & PSNR (dB) & Max.\ abs.\ pixel diff. & SSIM / LPIPS \\
\midrule
True severity channel & $31.30$ & -- & $0.949$ / $0.048$ \\
Zeroed (this audit)   & $31.28$ & $0$ & $0.944$ / $0.048$ \\
Dropped (this audit)  & $31.28$ & $0$ & $0.944$ / $0.048$ \\
\bottomrule
\end{tabular}
\end{table*}

\noindent \textbf{Result.} The two configurations agree exactly. PSNR, SSIM, and
LPIPS are identical to every reported decimal, and the maximum absolute pixel
difference over all $1{,}686$ test images is zero: the outputs are bitwise
identical. This is the algebraically expected outcome rather than a coincidence.
With the severity input set to zero, the trained severity columns contribute
$W_{\text{sev}}\!\cdot\!\mathbf{0}=\mathbf{0}$, the same contribution made by the
masked columns of the dropped configuration; the first convolution therefore
receives an identical tensor in both cases and every downstream layer follows.
The audit confirms that the implementation matches this algebra, with no
discrepancy from sampling or non-deterministic kernels.

\medskip
\noindent \textbf{Interpretation.} Both configurations reach $31.28$~dB,
$0.02$~dB below the full model ($31.30$~dB) and $0.05$~dB above the
schedule-matched baseline ($31.23$~dB). In this single exploratory run, most of \method{}'s oracle gain survives without the descriptor at inference, at an SSIM cost of $0.004$ relative to the baseline, consistent with Section~VII of the main text.

\medskip
\noindent \textbf{Sanity check.} The true-severity-map configuration measured
through this same pipeline reproduces $31.30$~dB, confirming the severity channel
reaches the network and that the audited measurements are meaningful rather than
artifacts of an inactive input.

\section*{Appendix F: Per-Video Cluster Table}
\label{app:pervideo}
This appendix lists the individual per-video mean PSNR difference
($\method - \text{DAB}$) and contributing frame count for each of the held-out
videos underlying the video-clustered significance analysis of
Section~VI-F (Fig.~4) in the main text. The aggregate statistics (Wilcoxon
$p=0.0002$; equal-video bootstrap $95\%$ CI $[0.044, 0.084]$~dB) are computed from exactly
this per-video data.

\begin{table}[h]
\centering
\caption{Per-Video Mean PSNR Difference and Frame Count, Sorted by Effect Size. The final row is the frame-weighted mean; the equal-video mean is $+0.065$~dB.}
\label{tab:pervideo}
\small
\begin{tabular}{lcc}
\toprule
Video & Mean $\Delta$ (dB) & $n$ (frames) \\
\midrule
1  & $-0.0033$ & $55$  \\
2  & $+0.0002$ & $40$  \\
3  & $+0.0168$ & $121$ \\
4  & $+0.0168$ & $72$  \\
5  & $+0.0452$ & $212$ \\
6  & $+0.0573$ & $33$  \\
7  & $+0.0647$ & $158$ \\
8  & $+0.0747$ & $59$  \\
9  & $+0.0819$ & $125$ \\
10 & $+0.0875$ & $143$ \\
11 & $+0.0973$ & $258$ \\
12 & $+0.0981$ & $160$ \\
13 & $+0.0996$ & $2$   \\
14 & $+0.1152$ & $185$ \\
15 & $+0.1156$ & $63$  \\
\midrule
\textbf{Total / mean} & \textbf{$+0.0721$} & \textbf{$1{,}686$} \\
\bottomrule
\end{tabular}
\end{table}

\noindent The $15$ test clusters correspond one-to-one with the $15$ held-out
source captures of the $36/15$ partition. Fourteen of fifteen favour \method{} (positive $\Delta$); video~1
shows a near-zero effect ($-0.0033$~dB) rather than a reversal, consistent with
the main text's description in Section~VI-F. Video~13 contributes only $2$
frames and is the smallest cluster in the dataset; excluding it changes the
video-level Wilcoxon test from $p=0.0002$ to $p=0.0004$ and the bootstrap
$95\%$~CI from $[0.044,0.084]$~dB to $[0.041,0.082]$~dB, so the result does not depend on the smallest cluster. Video identifiers are
listed as sequential integers; their mapping to source captures is given in the split manifest supplied with this submission.

\section*{Appendix G: Full Reproducibility and Training Configuration}
\label{app:repro}
This appendix consolidates all training-configuration details needed to
reproduce every table and figure in the main paper, extending the partial
hyperparameter list given in Section~VI.

\begin{table*}[t]
\centering
\caption{Full Training Configuration}
\label{tab:repro}
\small
\begin{tabular}{p{4.2cm}p{12.0cm}}
\toprule
Parameter & Value \\
\midrule
Optimizer & AdamW ($\beta_1=0.9$, $\beta_2=0.99$, no weight decay) \\
LR schedule & cosine from $2\times10^{-5}$ (minimum $10^{-6}$), period $34{,}284$ steps, stepped once per optimizer update; final rate $\approx1.05\times10^{-5}$ \\
EMA decay & $0.999$ \\
Training stages & DAB adaptation: $12$ epochs from GoPro-pretrained DiffIR; branch stage: $12$ further epochs from the DAB for every DiffIR arm, including the schedule-matched baseline; Restormer: $20$ epochs from GoPro-pretrained weights for both runs \\
Optimizer steps & $2{,}857$ micro-batches per epoch, accumulation $2$; the loss is divided by the accumulation factor under AMP gradient scaling, and the trailing micro-batch of each epoch is discarded, giving $1{,}428$ updates per epoch and $17{,}136$ per DiffIR branch \\
Batch size & physical $2$, gradient accumulation $2$, effective $4$; contrastive arm: physical $4$, no accumulation (same $1{,}428$ updates per epoch, final incomplete batch dropped) \\
Training patch size & $256\times256$, random crop \\
Test-time resolution & full-resolution, reflect-padded to a multiple of the
network stride ($8$) \\
Test-time tiling & $768$~px tiles, $48$~px overlap, weighted averaging \\
Evaluation metrics & PSNR and SSIM on 8-bit RGB (data range $255$, SSIM over colour channels, scikit-image); LPIPS with AlexNet \\
Mixed precision & AMP (\texttt{torch.amp}, \texttt{GradScaler}) \\
Sobel kernel & $3\times3$ \\
Gaussian smoothing & $15\times15$ kernel, $\sigma=2.6$ (OpenCV default for $\sigma=0$); gradient normalized by its maximum with $\epsilon=10^{-6}$ \\
Severity map storage & computed at full resolution, stored at $256\times256$ (area), bilinearly upsampled before cropping \\
Severity scale $\gamma(m_{\text{isp}})$ & $1.5-\tilde m\in[0.5,1.5]$, $\tilde m$ min--max normalized over all released metadata (training and test) \\

Contrastive temperature $\tau$ / weight & $0.1$ / $0.1$ \\
Frequency loss weight $\lambda_f$ & $0.1$ \\
FiLM inputs & five fields, z-score normalized with training-split statistics; mean $\ell_1$ auxiliary loss on the $384$-channel bottleneck, weight $0.05$ \\
Backbone & DiffIRS2~\cite{diffir}: \texttt{dim}$=48$, blocks $[3,5,6,6]$,
heads $[1,2,4,8]$, $4$ refinement blocks, $4$ diffusion timesteps \\
Widened layer & prior-encoder first projection, $48\to64$ input channels
(pixel-unshuffle$\times4$ of one severity channel), the $16$ added input channels $W'[:,48{:}64]$ zero-initialized \\
Added parameters & $9{,}216$ ($+0.034\%$ over $26.9$M) on DiffIR; $576$ on
NAFNet; $432$ on Restormer. DiffIR's conditioning encoder applies a
pixel-unshuffle before its first convolution, so the single severity channel
becomes sixteen unshuffled channels there, whereas NAFNet and Restormer read the
raw image and need one. The same structure makes DiffIR's added FLOPs
($+4.2$~G) larger than NAFNet's ($+1.06$~G). We report the measured figures
without attributing the ratio to a specific execution schedule \\
Train / test split & $5{,}714$ / $1{,}686$ paired images \\
Hardware & 1$\times$ NVIDIA Tesla T4 (16\,GB) \\
Software & PyTorch 2.x with CUDA 12.x and cuDNN 8.x, mixed-precision
(\texttt{torch.amp}); \texttt{BasicSR} for the DiffIR backbone \\
Seeds (runs 1--3) & $42$, $123$, $456$ \\
Checkpoint & final epoch for every reported run, including the shared DAB parent; no test data informs checkpoint choice or training length \\
Train / test videos & $36$ / $15$ disjoint source videos, from the released data files; version~1 of~\cite{iphoneblur} reports $39/12$, a documentation error \\
\bottomrule
\end{tabular}
\end{table*}

\section*{Appendix H: Extended Qualitative Results}
\label{app:qual}
Figure~5 of the main text shows the blurred input, the \method{}
restoration, and the ground truth, and notes that it omits the Domain-Adapted
Baseline column for column width. Figure~\ref{fig:dabquali} supplies that
comparison.

\begin{figure*}[t]
\centering
\includegraphics[width=\textwidth]{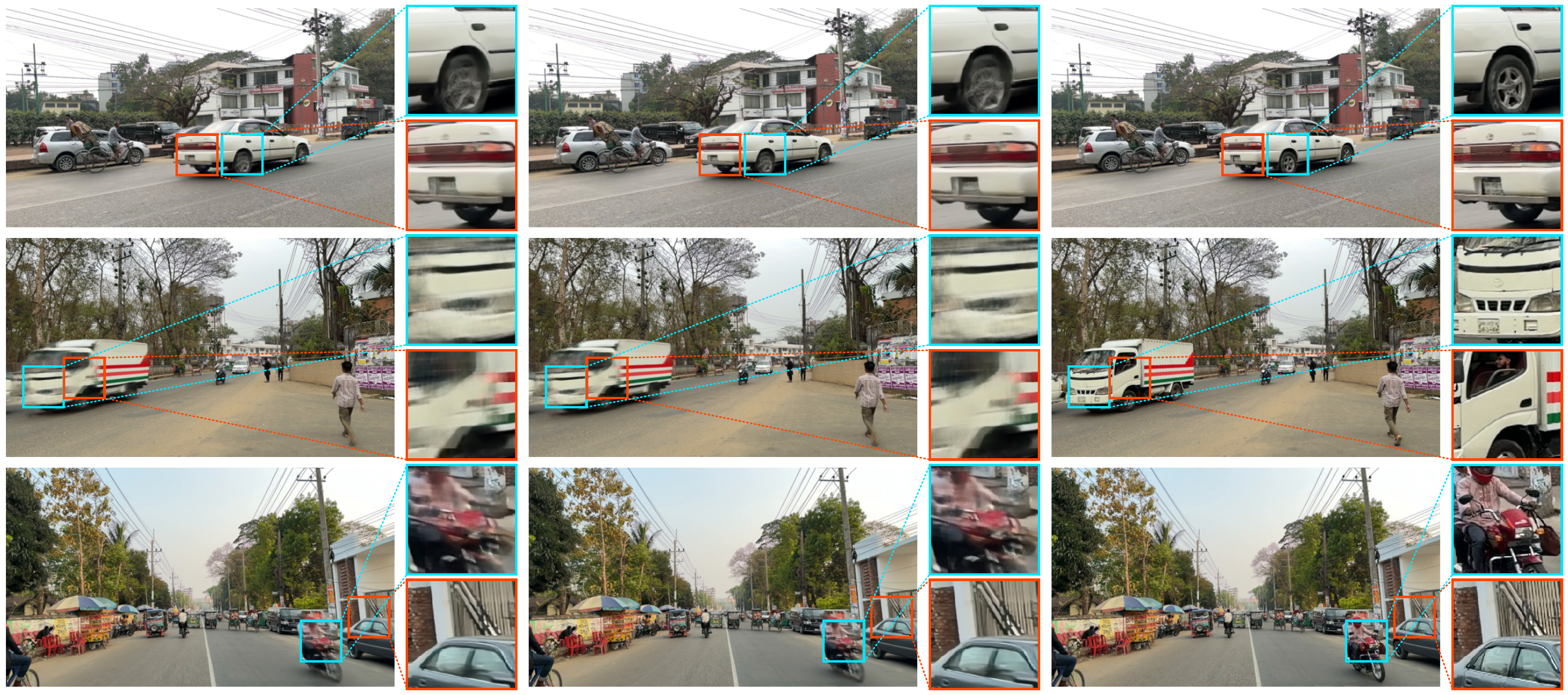}
\caption{Qualitative comparison with the DAB, omitted from Fig.~5 of the main text for width. Columns: DAB, \method{}, ground truth; two magnified high-motion crops per example (cyan, orange). Rows span the difficulty tiers.}
\label{fig:dabquali}
\end{figure*}

\medskip
\noindent\textbf{Failure cases.}
Figure~\ref{fig:failures} shows examples on which both models leave substantial
residual blur, illustrating the regime discussed in Section~VII of the
main text.

\begin{figure*}[t]
\centering
\includegraphics[width=\textwidth]{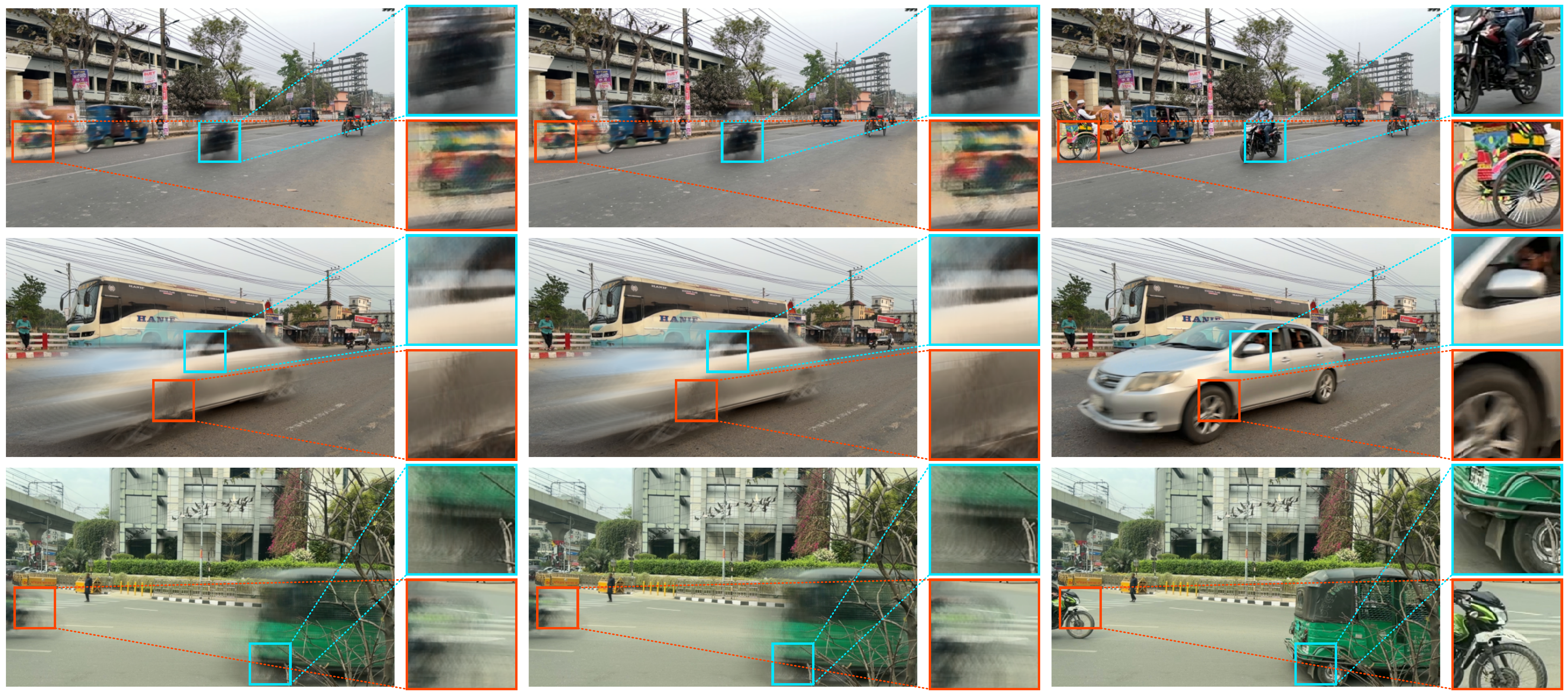}
\caption{Failure cases, in the layout of Fig.~\ref{fig:dabquali}. Blur is concentrated in a small, fast-moving foreground region against a sharp background, and neither model recovers the subject. A single per-image scale $\gamma(m_{\text{isp}})$ cannot represent such localized severity.}
\label{fig:failures}
\end{figure*}

\section*{Appendix I: Code and Data Availability}
\label{app:code}
Code sufficient to reproduce every table and figure in the main paper and this
supplement, including training configurations, evaluation scripts, and
per-image predictions with video identifiers and the $36/15$ source-video split manifest, will be released upon publication at
\url{https://github.com/C-loud-Nine/Severity-Channel-Conditioning}. The
iPhoneBlur benchmark used throughout is described in~\cite{iphoneblur}; its
licensing and consent terms for consumer-device captures are documented at
its accompanying publication~\cite{iphoneblur} (arXiv:2605.05990).

\end{document}